\documentclass{article} 

\usepackage[a4paper, total={7in, 10in}]{geometry}
\usepackage{lipsum}

\usepackage{biblatex}

\usepackage{amsthm}

\usepackage{amsmath,amssymb,amsfonts}
\usepackage{nameref}
\usepackage{textcomp}
\usepackage{algorithm}
\usepackage{xfrac}
\usepackage{faktor}
\usepackage{stmaryrd}
\usepackage{bm} 
\usepackage{hyperref} 
\newtheorem{Example}{Example}
\newtheorem{theorem}{Theorem}
\newtheorem{assumption}{Assumption}
\newtheorem{definition}{Definition}
\newtheorem{lemma}{Lemma}
\newtheorem{remark}{Remark}
\newtheorem{corollary}{Corollary}

\newtheorem{Illustration}{Illustration}

\newcommand\blfootnote[1]{%
  \begingroup
  \renewcommand\thefootnote{}\footnotetext{#1}%
  \addtocounter{footnote}{-1}%
  \endgroup
}

\begin{document}

\title{\textbf{Periodic Scenario Trees: A Novel Framework for Robust Periodic Invariance and Stabilization of Constrained Uncertain Linear Systems}}
\date{}

\maketitle

\begin{center}
\author{Yehia Abdelsalam,}
\author{Sankaranarayanan Subramanian,}
\author{Sebastian Engell.}
\end{center}

\blfootnote{The authors are with the Process Dynamics and Operations Group, Biochemical Engineering Department, TU Dortmund University, Dortmund, Germany.}
\blfootnote{Emails: yehia.abdelsalam@tu-dortmund.de, sankaranarayanan.subramanian@tu-dortmund.de, sebastian.engell@tu-dortmund.de.}
\blfootnote{This work was supported by TU Dortmund.}

\abstract{This work proposes a new a framework for determining robust periodic invariant sets and their associated control laws for constrained uncertain linear systems. Necessary and sufficient conditions for stabilizability by periodic controllers are stated and proven using finite step Lyapunov functions for the unconstrained case. We then introduce a scenario tree interpretation of finite step Lyapunov functions for uncertain systems and show that this interpretation results in useful criteria for the design of robust stabilizing controllers. In particular, novel convex feasibility criteria for the synthesis of simple static controllers and what we call linear interpolating tree periodic controllers with memory are derived. It is proven that for a sufficiently large length of the period, a stabilizing linear interpolating tree periodic controller can always be found using the proposed criterion provided that the uncertain system is stabilizable by such controllers. In this sense, the presented synthesis method is non-conservative. The results are then extended to constrained uncertain linear systems and conditions for controllers that realize robust periodic invariant sets which are less conservative than those that result from the known methods in the literature are derived.
}

\section{Introduction}
\label{Intro}
Obtaining non-conservative and tractable criteria for the stability analysis and controller synthesis for uncertain linear systems is a complex problem \cite{unsolved}. For tractability reasons, most of the criteria that are described in the literature are based on formulating the stability analysis and/or the controller synthesis as a feasibility problem of Linear Matrix Inequalities (LMIs) \cite{Boyd:94}. 
Earlier convex results for stability and stabilization of uncertain linear systems were obtained using a common quadratic Lyapunov function for the possible realizations of the uncertain system \cite{barmish_first,Barmish1985NecessaryAS,Gromel_convex}. 
In \cite{DAAFOUZ2001355}, results from \cite{DEOLIVEIRA1999261} were extended and a less conservative controller synthesis criterion was proposed using quadratic Parameter Dependent Lyapunov Functions (PDLFs) \cite{PDLF_first}, under the assumptions that the input matrix is perfectly known and time invariant and that the uncertain parameters can be measured or estimated in real-time hence are available each time step before the system evolves further. In \cite{PANDEY2017214}, the assumption that the input matrix is perfectly known and constant was removed, while the assumption that the uncertain parameters are real-time measurable was still made.
In \cite{ploynominal_LF_first}, homogeneous polynomial Lyapunov functions were proposed for the stability analysis of uncertain linear systems, and in \cite{CHESI20031027,CHESI2011621} necessary and sufficient convex conditions were derived using such functions, providing a non-conservative result for stability analysis of uncertain linear systems. 
Based on Lyapunov functions with non-monotonic terms \cite{non_monotonic_layp}, a method for controller synthesis was proposed in \cite{brazil} which similar to \cite{DAAFOUZ2001355,PANDEY2017214} assumes that the uncertain parameters can be perfectly measured or estimated in real-time.  
In \cite{LEE2006205,Lee_stab_LPV} a multi-step approach to the problem was proposed, and a non-conservative stability analysis as well as conditions for the synthesis of stabilizing controllers were derived, in which the proposed controllers were path dependent, i.e, depend on past values of the uncertain parameters. In \cite{PMLF_XIANG2018450}, the so called Parameter Memorized Lyapunov Functions (PMLFs) which are a multi-step generalization of PDLFs were introduced for the stability analysis of uncertain linear systems. In the context of Takagi-Sugeno (TS) \cite{Takagi-Sugeno-first} fuzzy representations of nonlinear systems, a multi-step approach was used to derive gain-scheduled control laws in \cite{Kruszewski_1}. Based on this result, sufficient conditions for obtaining periodic control laws were derived \cite{Guerra} for stabilizing a perfectly known nonlinear system that is represented by a TS fuzzy model.
A common feature in \cite{LEE2006205,Lee_stab_LPV,PMLF_XIANG2018450,Kruszewski_1,Guerra} is that not one quadratic function which is independent of the realization of the uncertainty is sought but for each path of uncertainty realizations, one of a number of quadratic functions is guaranteed to decrease, or in other words, an uncertainty dependent quadratic function is guaranteed to decrease. The disadvantage of such multi-step approaches is the exponential growth of the number of LMIs with the length of the path.  To reduce the exponential growth of the number of inequalities that result from PMLFs, it was shown in \cite{SALA_scenario} that given some inequalities that define a PMLF for the uncertain system, the feasibility of only a subset of these inequalities can be sufficient for guaranteeing stability. Heuristics were proposed for choosing a subset of inequalities to be checked. 
Some years ago, Finite Step Lyapunov Functions (FSLFs) \cite{Lazar_conv} were used for the stability analysis of perfectly known nonlinear systems \cite{Roman_Lazar_2,Lazar_3}. FSLFs do not necessarily decrease at each time step, but are guaranteed to decrease after a fixed finite number of steps (see \cite{Megretski,First_FSLF,pier_switches_FSLFS}). As will be shown, these will be very useful in our proposed criteria.

All the above results concern unconstrained systems. For constrained systems, the importance of invariant sets is well known \cite{bertsekas,BLANCHINI19991747,Kerrigan2000}. Ideally, one would like to find a control law that renders a set of states that is contained inside the constraints of the system of as large volume as possible invariant for the uncertain dynamics, and hence guarantees the satisfaction of the constraints at all time. The use of such sets as terminal constraints is very important in Model Predictive control (MPC) \cite{MAYNE2000789} especially for short prediction horizons, in which case the size of the terminal invariant set has a strong influence on the size of the feasible domain of the MPC. Such maximal invariant sets for linear systems are polytopic \cite{gilbert,Kolmanovsky1900} and can be characterized by a (possibly prohibitively large) number of vertices or hyperplanes. Ellipsoidal invariant sets on the other hand are used e.g. in \cite{KOTHARE19961361,ANGELI20083113,anil_kumar} due to their computational appeal as they can be defined using a single matrix. Unfortunately, ellipsoidal invariant sets can be very small or even non-existent for uncertain dynamics, even if the system is stabilizable. The concept of periodic set invariance \cite{First_quasi,canon_inv_1,canon_inv_2} is a relaxation of set invariance where it is allowed that the state of the system leaves the set under the condition that the state returns to the set again (see Figure \ref{Fig_ill_per_inv} for an illustration). In \cite{canon_inv_1,canon_inv_2}, ellipsoidal periodic invariant sets were used to enlarge the feasible domain of MPC. These sets were obtained by the use of periodic gains (an idea that was presented earlier in \cite{dinicolao_varying_K}). The concept of periodic invariance has regained interest lately in the design of MPC \cite{HANEMA2017137,output_periodic}. Periodic set invariance is closely related to FSLFs.
\subsection{Contributions}
The goal of this paper is to establish, for constrained uncertain linear systems, a flexible framework for finding robust periodic invariant ellipsoidal sets and their associated robust control laws, that reduces the conservatism that is present in the existing methods. 
This goal is achieved in several steps which are the main contributions of this paper: 
\begin{itemize}
    \item[1.]  
    Necessary and sufficient conditions for robust stabilization by periodic controllers are derived in Appendix \ref{section_periodic_stab} based on quadratic FSLFs. This result is helpful in proving the non-conservatism of the controller synthesis criteria that are proposed in section \ref{subsection_memory}.Since this result is only needed for the sake of proving the non-conservatism of the synthesis criteria but does not affect understanding the proposed synthesis methods, it is provided in Appendix \ref{section_periodic_stab}. 
    \item[2.] For unconstrained systems, a novel controller synthesis criterion is proposed. 
    It is based on a scenario tree interpretation of FSLFs for uncertain systems (see Figure \ref{Tree_first}) which is used to synthesize both static linear time invariant controllers and the novel Linear Interpolating Tree Periodic Controllers (LITPC). Such LITPCs have a finite memory which is used for storing a simulated scenario tree. It is shown by examples that the static controller synthesis that are proposed here can be less conservative than the periodic controller from \cite{canon_inv_2}. For the proposed LITPC, the conservatism decreases as the memory is increased, and by construction the criterion is always less conservative than the one in \cite{canon_inv_2} if they both have the same period.
    \item[3.] It is proven that for a sufficiently large memory (i.e., period), the LITPC synthesis becomes non-conservative, i.e., a stabilizing LITPC will always be found as long as the set of all stabilizing LITPCs is non empty. It should be noted that the necessity part of the proof of Lemma \ref{prop_suff_FSLF_K_different} is non-trivial and is illustrated in a simplified manner after the general proof in Illustration \ref{proof_illustration}. The non-standard proof uses a \textit{matrix replacement} Lemma (Lemma \ref{essential_lemma}) that is derived in Appendix \ref{Appendix_sec_Auxilary}.  
    \item[4.] For constrained uncertain linear systems, the results are extended from the unconstrained case to the constrained case and are used to construct robust periodic invariant ellipsoidal sets. Figure \ref{new_Fig_ill_per_inv} illustrates the periodic invariant ellipsoids that result from the scenario tree framework which is used in this work.
    Unlike existing results, the action of the periodic controllers that are proposed here depends on both the time step within the period as well as the location of the state with respect to the simulated and stored scenario tree. 
    This can result in a large periodic invariant set even for short periods. 
    Due to the design of both the static controllers and the LITPCs using scenario trees, even by using the static controller synthesis criterion one can obtain larger robust periodic invariant sets than the ones from \cite{canon_inv_2}, while the sets obtained by the LITPC synthesis criteria are necessarily larger than or equal to the ones from \cite{canon_inv_2}. Such larger periodic invariant sets are of great interest for the design of MPC (see Remark \ref{remark_MPC_large} below). 
\end{itemize}

The proofs of the theorems, lemmas and corollaries are provided in Appendix \ref{Appendix_proofs}.
Throughout, the advantages of the new methods over the existing results are demonstrated using numerical examples. The semi-definite programs are formulated using YALMIP \cite{Yalmip} and solved using SDPT3 \cite{SDPT3} or MOSEK \cite{mosek}.

\subsection{Progress Over Existing Results}
As mentioned in the previous section, a large number of results on stability and stabilization as well as (periodic) set invariance of uncertain or Linear Parameter Varying (LPV) systems exist in the literature. Some of these results are also using multi-step approaches \cite{LEE2006205,Lee_stab_LPV,PMLF_XIANG2018450,Kruszewski_1,Guerra,SALA_scenario,canon_inv_1,canon_inv_2}. 

The main novelty of the framework that is proposed in this paper in comparison to all the above mentioned multi-step approaches is the periodic use of a scenario tree to determine the control law as detailed in Algorithm \ref{algorithm_controller_periodic}. Using the proposed periodic scenario tree framework for the controller synthesis in conjunction with quadratic FSLFs facilitates the determination of robust periodic invariant ellipsoids which are by construction less conservative than existing methods. In addition to this major novelty, a major difference to the work \cite{Kruszewski_1,Guerra}, is that the work in \cite{Kruszewski_1,Guerra} is assuming a gain scheduling setting, i.e., the uncertain parameters are assumed to be known at the current time step, while this work considers the robust setting, i.e., at the current time step the uncertain parameters are assumed to be unknown, but the effect of the past uncertainties is considered via the knowledge of the current state. The work in \cite{brazil} also considers a gain scheduling setting and unconstrained systems and not the robust setting for constrained systems as it is addressed here. Another key difference to the work in \cite{Kruszewski_1,Guerra,LEE2006205,Lee_stab_LPV} is that the novel framework results in the computation of robust periodic invariant ellipsoids (as shown in section \ref{section_constrained}), which is the main motivation of this paper. This is not the case for \cite{Kruszewski_1,Guerra,LEE2006205,Lee_stab_LPV}. As mentioned in the previous section, in \cite{LEE2006205,Lee_stab_LPV,PMLF_XIANG2018450,Kruszewski_1,Guerra}, for each path of uncertainty realizations one of many quadratic functions is guaranteed to decrease, or in other words, an uncertainty dependent quadratic function is guaranteed to decrease. In the work presented here the quadratic function used at the beginning and the end of the period is constructed with the same positive definite matrix $P_0$ (i.e., the quadratic Lyapunov function does not depend on the uncertainty), which is not the case in \cite{Kruszewski_1,Guerra,LEE2006205,Lee_stab_LPV}. Even though the same matrix $P_0$ is used at the beginning and the end of the period, it is proven that if the system is stabilizable by LITPC then our synthesis method will find a LITPC with a period length $N$ which is the same as the period length of the quadratic FSLF. This can be seen from the necessity and the sufficiency that are proven in Theorem \ref{suff_theorem_differnt_gains}.

Unlike the work presented here, the work in \cite{PMLF_XIANG2018450,SALA_scenario} only deal with stability analysis and not with controller synthesis.

The results that are related the most to the work presented here have emerged from the MPC community (see \cite{canon_inv_1,canon_inv_2}), which consider a robust constrained setting, and are concerned with determining periodic invariant ellipsoidal sets as well. The major difference between our work and \cite{canon_inv_1,canon_inv_2} is illustrated in Figures \ref{Fig_ill_per_inv} and \ref{new_Fig_ill_per_inv}. Figure \ref{Fig_ill_per_inv} illustrates the robust periodic invariant sets that result from the controllers in \cite{canon_inv_1,canon_inv_2}, while Figure \ref{new_Fig_ill_per_inv} illustrates the periodic invariant sets that result from the LITPC in this work. This major difference results from using the scenario tree structure of the controller in our work. Consequently, the periodic invariant sets that result from the LITPC approach presented here are less conservative than the ones that result from \cite{canon_inv_1,canon_inv_2}. 

Thus, to the best of our knowledge, by construction the proposed LITPC synthesis method produces robust periodic invariant ellipsoids which are larger than or equal to the methods in existing literature (see also Example \ref{example_sets}).

\subsection{Notations} 
The set of real numbers is denoted by $\mathbb{R}$. Let $a$ be some integer. The set of integers greater than or equal to $a$ is denoted by $\mathbb{Z}_{\geq a}$. The Euclidean norm of a vector $x$ is denoted by $\|x \|$. For two vectors $a \in \mathbb{R}^n$, $b\in \mathbb{R}^n$, $a<b$ ($a\leq b$) means that the inequality holds elementwise. For a matrix $A \in \mathbb{R}^{n \times n}$, $\| A \|$ denotes the induced $\ell^2$ matrix norm.  For a square matrix $P$, $P>0$ ($P\geq 0$) denotes that $P$ is positive definite (positive semi-definite). For two square matrices of the same dimension, $P_1>P_2$ ($P_1 \geq P_2$) means that $P_1-P_2>0$ $(P_1-P_2\geq 0)$. The convex hull operation is denoted by $Co(\cdot)$.  
For two positive integers $c_1$ and $c_2$, the remainder of dividing $c_1$ by $c_2$ is denoted by $[c_1 \mod c_2]$, while the quotient is denoted by $[\sfrac{c_1}{c_2}]$. The square identity matrix is denoted by $\mathbf{I}$, where the dimension is inferred from the context. 
\subsection{System Description and Preliminaries}
We consider discrete-time linear systems of the form:
\begin{equation} \label{system_dynamics}
x_{t+1}=A_t x_t+B_t u_t,  \text{  }  \text{  }  \text{  } t\in \mathbb{Z}_{\geq 0}
\end{equation}
where $x_t\in \mathbb{R}^{n_x}$, $u_t\in \mathbb{R}^{n_u}$ denote the system state and input at time step $t$. Let 
$\Gamma=\{ 1,2,\dots, n_d \}$. 

Define $\mathcal{D}=\{ (\bar{A}_i,\bar{B}_i),\textbf{ }\forall i\in \Gamma\}$, where $\bar{A}_i$ and $\bar{B}_i$, $\forall i\in \Gamma$ are known matrices.
Define the compact set $\mathbf{D}$ as $\mathbf{D}=Co(\{ (\bar{A}_i,\bar{B}_i), \textbf{ }i \in \Gamma \})$, i.e., $\mathbf{D}=Co(\mathcal{D})$. 
Note that according to our notation, $(A_1,B_1)$ means the actual realization of the dynamics at $t=1$, which should not be confused with $(\bar{A}_1, \bar{B}_1)$ which is the first element of the set of vertices $\mathcal{D}$.

\begin{assumption}\label{Assumption_1}
The matrices $A_t$ and $B_t$ satisfy:
\begin{itemize}
    \item [A.] $(A_t,B_t)=\sum^{n_d}_{i=1}\alpha_{t,i} (\bar{A}_i,\bar{B}_i)$, $\forall t\in \mathbb{Z}_{\geq 0}$, $\alpha_{t,i}\geq 0$, $\forall t\in \mathbb{Z}_{\geq 0}$, $\forall i\in \Gamma$, where at each $t\in \mathbb{Z}_{\geq 0}$, $\sum^{n_d}_{i=1}\alpha_{t,i}=1$. 
    \item[B.] $(A_t,B_t)$ can vary arbitrarily inside $\mathbf{D}$. 
    \item[C.] $(A_t,B_t)$ and correspondingly $\alpha_{t,i}$, $\forall i\in \Gamma$ are unknown at $t$.
    \item[D.] The state $x_t$ is known at $t$. 
\end{itemize}
\end{assumption}
 Note that Assumption \ref{Assumption_1}.A can be succinctly written as,  $(A_t,B_t)\in\mathbf{D}$, $\forall t\in \mathbb{Z}_{\geq 0}$, and the vector 
 \begin{equation*}
 \alpha_t=\begin{pmatrix}
     \alpha_{t,1} &\alpha_{t,2}& \dots &\alpha_{t,n_d}
 \end{pmatrix}^T
 \end{equation*}
 completely and uniquely determines $(A_t,B_t)$ at time $t$.

A sequence of realizations of the system and input matrices that occur until the current time step $t$ is defined as 
\begin{equation*}
    \textbf{d}^t=\{(A_0,B_0),(A_1,B_1),\dots\, (A_{t-1},B_{t-1})\},
\end{equation*}
where $\textbf{d}^t\in \mathbf{D}^t$. 

From hereon, we assume that the system is controlled by some control law $\kappa$. 

The closed-loop system is denoted by $\mathcal{S}$.
\begin{definition}\label{stability_def}
The closed-loop system $\mathcal{S}$ is called robustly exponentially stable on a subset $\mathbf{X} \subseteq \mathbb{R}^{n_x}$ if there exist scalars $c\geq 1$ and $\lambda \in (0,1)$ such that $\|x_t\|\leq c \lambda^t \|x_0\|$, $\forall x_0 \in \mathbf{X}$, $\forall \textbf{d}^t\in \mathbf{D}^t$ and $\forall t \in \mathbb{Z}_{\geq0}$. 
\end{definition}

Consider positive definite quadratic functions of the form 
\begin{equation}\label{global_FSLF_eqn}
    V(x)=x^TP_0x,
\end{equation}
where $P_0>0$ is a positive definite symmetric matrix. 
 
Let $\mathbf{X} \subseteq \mathbb{R}^{n_x}$. Assume that $\mathbf{X}$ contains the origin in its interior. The following definition is adapted from \cite{Lazar_conv}. 
\begin{definition}\label{Lyap_def}
 A function $V:\mathbf{X} \to\mathbb{R}$ of the form \eqref{global_FSLF_eqn} is called a Finite Step Lyapunov Function (FSLF) for the closed-loop system on $\mathbf{X}$ if there exists $N \in \mathbb{Z}_{\geq 1}$ such that $\forall x_0 \in \mathbf{X} \setminus \{0\}$, $V(x_{(m+1)N})<V(x_{mN})$, $\forall m \in \mathbb{Z}_{\geq 0}$, for any $x_{mN}\neq 0$, $\forall \textbf{d}^{(m+1)N}\in \mathbf{D}^{(m+1)N}$. We call $N$ the period of the FSLF.
\end{definition}

\section{Novel Controller Synthesis}\label{Section_controller_synthesis}
The evolution of the closed-loop uncertain system $\mathcal{S}$ according to the vertices $(\Bar{A}_i, \Bar{B}_i)$ of the set $\mathbf{D}$ will define a switched system. We will denote that switched system by $\mathcal{S}_{\mathcal{D}}$. 
The evolution of $\mathcal{S}_{\mathcal{D}}$ can be modeled by a scenario tree as shown in Figure \ref{Tree_first}. 

We will consider only the first $N$ time steps ($N+1$ samples) of the evolution, $t \in \{0,1,\dots,N \}$, and hence we will consider only scenario trees of finite length $N$, as these will suffice to prove the stability of the closed-loop system (as will be shown in details).
The set of indices $(j,t)$ of the tree states $x^j_t$ and inputs $u^j_t$ occurring at time step $t=l$ is defined as $I_{l}:=\{(j,t)|\textbf{ }t=l,\textbf{ }j\in \{1,2,\dots,n_d^{l} \} \}$. The notation $I_{ \llbracket l_1,l_2 \rrbracket}$ will denote the set of indices occurring from time step $l_1$ until time step $l_2$, i.e., $I_{\llbracket l_1,l_2 \rrbracket}:=\bigcup_{l=l_1}^{l_2} I_l$. Note that $x^1_0=x_0$ and $u^1_0=u_0$.
The dynamics of $\mathcal{S}_{\mathcal{D}}$ is defined by,
\begin{equation}\label{dynamics_tree_1}
    x^j_{t+1}=\bar{A}_{i^j_{t+1}} x^{p(j)}_t+\bar{B}_{i^j_{t+1}} u^{p(j)}_t,  \textbf{ } \forall (j,t+1)\in I_{\llbracket 1,N \rrbracket}, 
\end{equation}
where $(p(j),t)$ is the index of the parent of node $(j,t+1)$, $i^j_{t+1}$ denotes the index $i$ of the element from the set $\mathcal{D}$ at step $t$, that resulted in the child node of index $(j,t+1)$. 

At the beginning of each period, i.e., every $N$ time steps, a controller $\kappa_{\mathcal{D},N}$ acts along the scenario tree (i.e., $\mathcal{S}_\mathcal{D}$) of length $N$ from the measured state $x_t$, where the root node of the tree is set to $x^1_0=x_t$ (see Figure \ref{Tree_first}). Hence, $\kappa_{\mathcal{D},N}$ determines $u^j_{t}$, $\forall (j,t)\in I_{\llbracket 0,N-1 \rrbracket}$. This means that at the beginning of each period, the controller predicts the
evolution of the controlled switched system $\mathcal{S}_\mathcal{D}$, for the upcoming $N$ time steps, and stores the result in terms of inputs and states $u^j_t$, $\forall (j,t)\in I_{\llbracket 0,N-1 \rrbracket}$, $x^j_t$, $\forall (j,t)\in I_{\llbracket 1,N \rrbracket}$. The prediction of the evolution along the full scenario tree does not take place at each time step but only every $N$ time steps. If $(A_t,B_t)$ are always equal to one of the elements in the set $\mathcal{D}=\{ (\bar{A}_i,\bar{B}_i),\textbf{ }\forall i\in \Gamma\}$, then the state will always be one of the nodes of the stored scenario tree $x^j_t$, and the input applied to the plant will always be equal to the corresponding stored input $u^j_t$. 
\begin{remark}
The controllers that we propose belong to the class of (not necessarily optimal) robustly stabilizing controllers of uncertain linear systems. The scenario tree formulation of the problem is inspired from the domains of stochastic programming \cite{BirgLouv97} and robust min-max and multi-stage MPC \cite{Scokaert_Mayne,dela_pena,lucia2013,sankar}.  
\end{remark}
\textbf{Main Idea:} The main idea is to determine a controller $\kappa_{\mathcal{D},N}$ offline that results in the reduction of the quadratic function $x^TP_0x$ on $N$ time steps for all possible scenarios of the scenario tree (see Figure \ref{Tree_first}). Since the controller $\kappa_{\mathcal{D},N}$ is only defined for the switched system $\mathcal{S}_\mathcal{D}$, an interpolating stabilizing controller $\kappa_N$ will in the next step be determined from $\kappa_{\mathcal{D},N}$ for the actual uncertain system $\mathcal{S}$.\\
\textbf{Problem Summary:} The key step in our approach is to find the controller $\kappa_{\mathcal{D},N}$. Specifically, for the system $\mathcal{S}_{\mathcal{D}}$, we want to find a controller $\kappa_{\mathcal{D},N}$ with period $N$ and a symmetric matrix $P_0>0$ such that $\forall x_0 \in \mathbb{R}^{n_x} \setminus \{ 0 \}$, $x^{j^T}_N P_0 x^j_N<x^T_0 P_0 x_0$, $\forall (j,N)\in I_{N}$. This implies robust exponential stability of $\mathcal{S}_\mathcal{D}$ (see Theorem \ref{suff_theorem_fixed_K} and Corollary \ref{corollary_necc_suff_imp} in Appendix \ref{section_periodic_stab}). We then will show that $\kappa_N$ as defined by Algorithm \ref{algorithm_controller_periodic} stated below robustly exponentially stabilizes $\mathcal{S}$.
\begin{definition}\label{Def_GTPC}
    \textbf{Interpolating Tree Periodic Controller (ITPC):} An ITPC is a controller that computes its control actions according to Algorithm \ref{algorithm_controller_periodic}.
\end{definition}

\begin{definition}\label{beta_definition}
Define $\beta^j_{t+1}=\alpha_{t,i}\beta^{p(j)}_{t}$, $\forall (j,t+1) \in I_{\llbracket 1,N \rrbracket}$, where $\alpha_{t,i}$ $\forall i \in \Gamma$ satisfy Assumption \ref{Assumption_1}, with the initial condition $\beta^1_0=1$. This implies that for each $(j,t) \in I_{\llbracket 0,N \rrbracket}$, $\beta^j_t\geq 0$ and for each $t \in \{0,1,\dots,N \}$, $\sum^{n^t_d}_{j=1}\beta^j_{t}=1$.
\end{definition}

\begin{algorithm}
\caption{Online Implementation: ITPC}
\begin{tabular}{p{1.8 cm} p{12.5cm}}\label{algorithm_controller_periodic}
\textbf{Require} &  The set $\mathcal{D}$ and the controller $\kappa_{\mathcal{D},N}$ (this is an offline step).\\
\textbf{M} & Measure $x_t$.\\
\textbf{Sim./Store} & \textbf{If $[t \mod N]=0$:} Use $\kappa_{\mathcal{D},N}$ to simulate $\mathcal{S}_{\mathcal{D}}$: \\
Step S1& Set $x^1_0=x_t$. Simulate $\mathcal{S}_{\mathcal{D}}$, i.e., compute $u^j_k$ and $x^j_k$, $\forall (j,k) \in I_{\llbracket 0,N \rrbracket}$.\\
Step S2&Store the computed evolution of $\mathcal{S}_{\mathcal{D}}$.\\ 
Step S3& Apply $u_t=u^1_0$ to the plant.\\
Step S4& Set $t=t+1$. Go to \textbf{M} at the next time step.\\
\textbf{Interpolate} & \textbf{If $[t \mod N]\neq0$:} Compute the result of $\kappa_N$:\\
Step I1& Set $k=[t \mod N]$. Compute $\beta^j_k$ from the measured $x_t$ and the stored $x^j_k$ by solving $x_t=\sum^{n^k_d}_{j=1}\beta^j_kx^j_k$, $\sum^{n^k_d}_{j=1}\beta^j_k=1$, $\beta^j_k\geq 0$ .\\
Step I2& Compute $u_t$ from $\beta^j_k$ and the stored $u^j_k$ using $u_t=\sum^{n^k_d}_{j=1}\beta^j_ku^j_k$.\\
Step I3& Apply the computed $u_t$ to the plant.\\
Step I4& Set $t=t+1$. Go to \textbf{M} at the next time step.\\
\end{tabular}
\end{algorithm}

It is important to note that unlike in multi-stage MPC \cite{Scokaert_Mayne,dela_pena}, the prediction of the scenario tree is not based on the solution of an online optimization problem, but is based on a fixed control law $\kappa_{\mathcal{D},N}$ that is computed offline. How $\kappa_{\mathcal{D},N}$ is computed will be detailed in the next subsections. Only Algorithm \ref{algorithm_controller_periodic} is executed online.

The following Lemma characterizes the properties of the closed-loop resulting from the ITPC defined by Algorithm \ref{algorithm_controller_periodic}.
\begin{lemma}\label{convexity_lemma}
Consider the time steps $t \in \{0,1,\dots, N \}$. Assume that the system  $\mathcal{S}$ is controlled using ITPC defined in Algorithm \ref{algorithm_controller_periodic} with period $N$. 
The following hold for $\mathcal{S}$:
 \begin{itemize}
     \item [A.] The state $x_t$ of the system $\mathcal{S}$ satisfies $x_t=\sum_{j=1}^{n_d^t} \beta^j_{t} x^j _t$, $\forall t\in \{0,1,\dots, N \}$, i.e., $x_t \in Co({x^j_t}, \textbf{ } j \in \{ 1,2,\dots, n_d^t\})$.
     \item [B.]  A function $x^T P_0 x$ is a FSLF of period $N$ for $\mathcal{S}$, if and only if it is a FSLF of period $N$ for $\mathcal{S}_{\mathcal{D}}$.
 \end{itemize}
\end{lemma}
\begin{proof}
See Appendix \ref{Appendix_proofs}.
\end{proof}
\begin{remark}
In the proposed formulation, the inputs that are applied at the same tree node are the same, independent of the uncertainty that materializes thereafter. This a key difference from previous results that we have mentioned in section \ref{Intro} (see \cite{DAAFOUZ2001355,PANDEY2017214,brazil,Kruszewski_1,Guerra}), where it is assumed that the uncertainty is real-time measurable and the control law is designed by taking the current uncertainty into account. In other words, in our formulation the controller reacts to the past realizations of uncertainty but not the current and future uncertainties since these are unknown. This is called non-anticipativity in stochastic optimization and robust MPC where causality is respected.
\end{remark}
\begin{remark}
At time step $t$, the computation of the control inputs $u_t$ is based on the values of $\beta^j_{t}$, $\forall j \in \{1, 2,\dots,n_d^t\}$ (see Definition \ref{beta_definition}). Note that $\beta^j_{t}$, $\forall j \in \{1, 2,\dots,n_d^t\}$ summarize the information about the past uncertainties until $t-1$, i.e., $\alpha_k$, $k \in \{0,1,\dots,t-1 \}$, and do not contain any information about the current uncertainty $\alpha_t$. We also assume that we only measure the state of the plant $x_t$ and not $\beta^j_{t}$. The coefficients $\beta^j_{t}$, $\forall j \in \{1, 2,\dots,n_d^t\}$ are determined from the linear system of equations given in Step I1 of Algorithm \ref{algorithm_controller_periodic}. Note that there always exists a solution to this system of equations, which however might not be unique. Nonetheless, the result of Lemma \ref{convexity_lemma} holds for any of these non-unique solutions.
\end{remark}
\begin{figure}
\begin{center}
\includegraphics[width=\columnwidth]{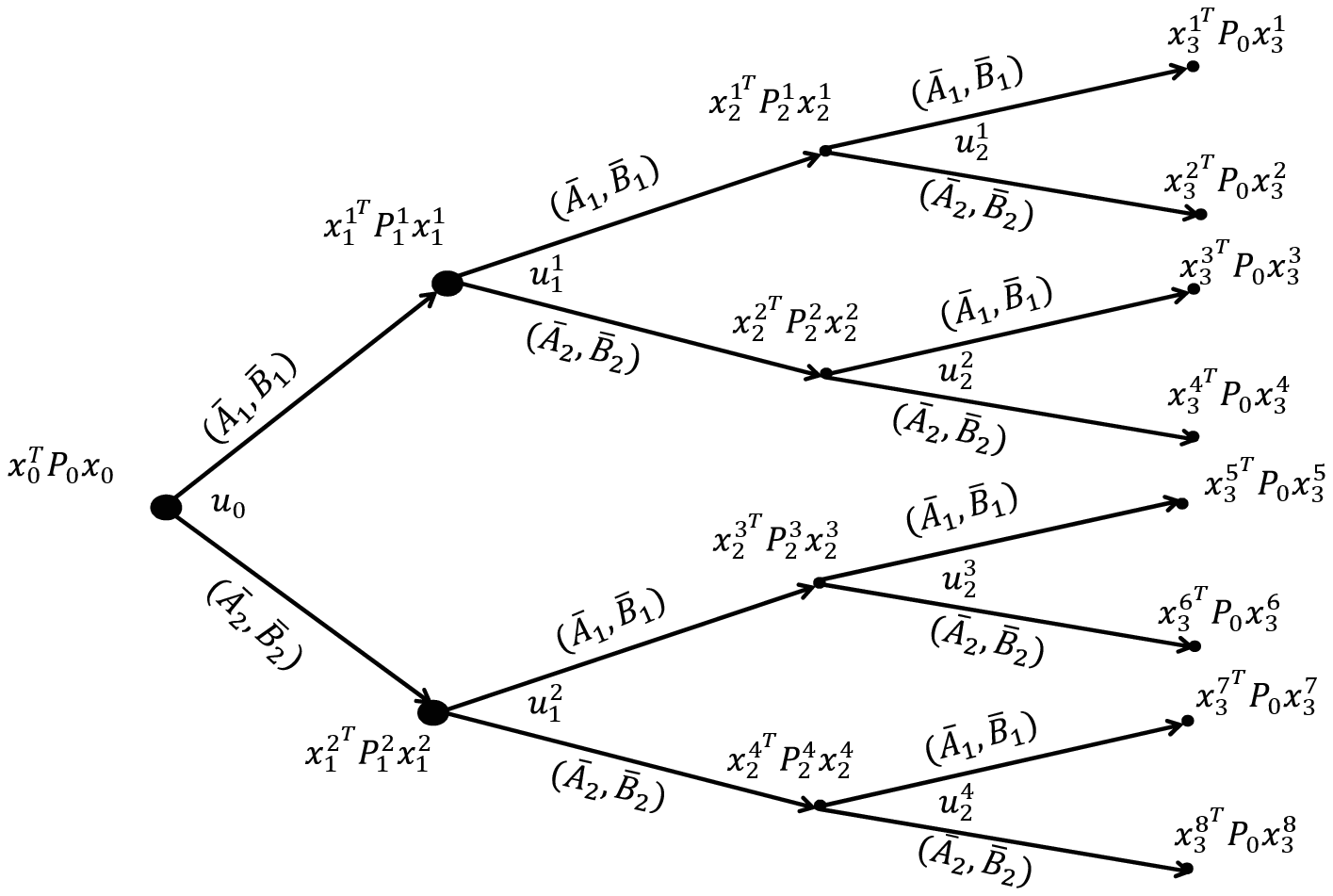} 
\caption{Scenario tree representation of $\mathcal{S}_{\mathcal{D}}$ for three time steps ($N=3$). The set $\mathbf{D}$ has two vertices, i.e., $\Gamma=\{1,2\}$. The quadratic functions at each stage (time step) are determined such that they are smaller than the quadratic function at their parent node. For the last stage ($t=N$), the symmetric positive definite matrix $P_0$ which was used for the initial state $x_0$ is used again.}
\label{Tree_first}  
\end{center} 
\end{figure}
\subsection{Synthesis of Static Controllers}\label{subsection_static_A}
We first consider static controllers of the form $\kappa_N(x_t)=Kx_t$. Note that this can be considered as a simple special case of the ITPC defined by Algorithm \ref{algorithm_controller_periodic} with $\kappa_{\mathcal{D},N}(x_t)=\kappa_N(x_t)=Kx_t$, where the scenario tree neither needs to be simulated nor stored. The closed-loop dynamics \eqref{dynamics_tree_1} then is defined by
\begin{equation}\label{propagation_eqn_1} 
    x^j_{t+1}=(\bar{A}_{i^j_{t+1}}+\bar{B}_{i^j_{t+1}} K) x^{p(j)}_t, \textbf{ } \forall (j,t+1)\in I_{\llbracket 1,N \rrbracket}. 
\end{equation}

 We want to find a gain matrix $K$ such that there exists a symmetric positive definite matrix $P_0$, such that $\forall x_0 \in \mathbb{R}^{n_x} \setminus \{ 0\}$, we have $x^{j^T}_N P_0 x^j_N<x^T_0 P_0 x_0$, $\forall (j,N)\in I_{N}$  for some $N \in \mathbb{Z}_{\geq 1}$. 
\begin{lemma}\label{prop_suff_FSLF_K_fixed}
    The system $\mathcal{S}$ is robustly exponentially stabilized by the control law $\kappa(x_t)=Kx_t$ if there exist $N \in \mathbb{Z}_{\geq 1}$ and symmetric matrices $P^j_t >0$ such that:  
\begin{equation}\label{P_static_1}
    P^{p(j)}_t-(\bar{A}_{i^j_{t+1}}+\bar{B}_{i^j_{t+1}} K)^T P^j_{t+1} (\bar{A}_{i^j_{t+1}}+\bar{B}_{i^j_{t+1}} K)> 0
\end{equation}
$\forall (j,t+1)\in I_{\llbracket 1,N-1 \rrbracket}$, and 
\begin{equation}\label{P_static_2}
     P^{p(j)}_{N-1}- (\bar{A}_{i^j_{N}}+\bar{B}_{i^j_{N}} K)^T P_0 (\bar{A}_{i^j_{N}}+\bar{B}_{i^j_{N}} K)> 0
\end{equation}
$\forall (j,N)\in I_{N}$, where $P^1_0=P_0$. 
\end{lemma}
\begin{proof}
See Appendix \ref{Appendix_proofs}.  
\end{proof}

The following Theorem provides a sufficient convex criterion for determining such a stabilizing static controller for $\mathcal{S}$.
\begin{theorem} \label{Theorem_stability_static_suff}
    If there exist $N \in \mathbb{Z}_{\geq 1}$, a matrix $G \in \mathbb{R}^{n_x \times n_x}$, a matrix $L\in \mathbb{R}^{n_u \times n_x}$ and symmetric positive definite matrices $S^j_t \in \mathbb{R}^{n_x \times n_x}$, $(j,t)\in I_{\llbracket 0,N-1 \rrbracket}$, where $S^1_0=S_0$ such that 
\begin{equation}\label{LMI_1_static}
\begin{pmatrix}
G^T+G-S^{p(j)}_t &  \textbf{ } &  G^T \bar{A}^T_{i^j_{t+1}}+L^T \bar{B}^T_{i^j_{t+1}}\\
\bar{A}_{i^j_{t+1}} G + \bar{B}_{i^j_{t+1}} L  & \textbf{ } &  S^j_{t+1}
\end{pmatrix}> 0
\end{equation}
$\forall (j,t+1)\in I_{\llbracket 1,N-1 \rrbracket}$   
\begin{equation}\label{LMI_2_static}
\begin{pmatrix}
G^T+G-S^{p(j)}_{N-1} &  \textbf{ } \textbf{ } \textbf{ }  \textbf{ }&  G^T \bar{A}^T_{i^j_{N}}+L^T \bar{B}^T_{i^j_{N}}\\
 \bar{A}_{i^j_{N}} G + \bar{B}_{i^j_{N}} L  &  \textbf{ }  \textbf{ } \textbf{ }  \textbf{ } &  S_0
\end{pmatrix}> 0
\end{equation}
$\forall (j,N)\in I_{N}$, then the closed-loop system $\mathcal{S}$ with the control law $\kappa(x_t)=K x_t$ is robustly exponentially stable, where 
\begin{equation}\label{K_eqn}
    K=L G^{-1}.
\end{equation} 
\end{theorem}
\begin{proof}
    See Appendix \ref{Appendix_proofs}.
\end{proof}
\begin{Example}\label{Example_1}
Consider the system \eqref{system_dynamics} with $A_t=\begin{pmatrix}
a^1_t & 1.12 \\ a^2_t & a^3_t
\end{pmatrix}$ and $B_t=\begin{pmatrix}
1\\1
\end{pmatrix}$, where $a^1_t\in [0.2,1.8]$, $a^2_t \in [-0.35,0.35]$, $a^3_t \in [0.1,1.9]$. The quadratic stabilizability criterion for this system fails to find a feasible solution. Also, applying the method in \cite{canon_inv_2} to find periodic stabilizing gains failed. Note that since the system is unconstrained, only the inequalities (18) in \cite{canon_inv_2} are used. On the other hand, by solving the LMIs given in Theorem \ref{Theorem_stability_static_suff} with $N=2$, a feasible solution is found, for which $K=\begin{pmatrix}
-0.8956 & -1.103 
\end{pmatrix}$ is the resulting stabilizing gain matrix.
\end{Example}
Note that the result of Example \ref{Example_1} can not be generalized, i.e., there are also examples where the method in \cite{canon_inv_2} finds a stabilizing sequence of feedback gains, while our method cannot find a static stabilizing feedback gain.  
\subsection{Synthesis of Linear Interpolating Tree Periodic Controllers with Memory}\label{subsection_memory} 
In this section, we propose and prove new convex necessary and sufficient conditions for the stabilizability of $\mathcal{S}$. In particular, our criterion results in the determination of a non-conservative stabilizing linear interpolating tree periodic memory based control law as defined by Algorithm \ref{algorithm_controller_periodic}.  

\begin{definition}\label{Def_Memory_controller}
\textbf{Linear Interpolating Tree Periodic Controller (LITPC):} A LITPC is an ITPC (see Definition \ref{Def_GTPC}) according to Algorithm \ref{algorithm_controller_periodic} in which the inputs of the simulated scenario tree (that result from the controller $\kappa_{\mathcal{D},N}$) are $u^j_k=K^j_k x^j_k$, $\forall (j,k) \in I_{\llbracket 0,N-1 \rrbracket}$ with $K^1_0=K_0$. 
\end{definition}
For the first $N$ time steps, the evolution of $\mathcal{S}_\mathcal{D}$ is given by
\begin{equation}\label{propagation_eqn_2_diff_k}
    x^j_{t+1}=(\bar{A}_{i^j_{t+1}}+\bar{B}_{i^j_{t+1}} K^{p(j)}_t) x^{p(j)}_t, \textbf{ } \forall (j,t+1)\in I_{\llbracket 1,N \rrbracket}. 
\end{equation}

We want to find gains $K^j_t$, $\forall (j,t)\in I_{\llbracket 0,N-1 \rrbracket}$, such that there exists a symmetric positive definite matrix $P_0$, such that $\forall x_0 \in \mathbb{R}^{n_x} \setminus \{ 0\}$, $x^{j^T}_N P_0 x^j_N<x^T_0 P_0 x_0$, $\forall (j,N)\in I_{N}$ for some $N \in \mathbb{Z}_{\geq 1}$. 

 The following lemma provides necessary and sufficient conditions for the stabilizability of $\mathcal{S}$ by LITPC.
\begin{lemma}\label{prop_suff_FSLF_K_different}
    The system $\mathcal{S}$ is robustly exponentially stabilizable by an LITPC if and only if there exists some  $N \in \mathbb{Z}_{\geq 1}$, for which there exist symmetric matrices $P^j_t >0$ and an LITPC of period $N$ whose gains $K^j_t$ result in: 
\begin{equation}\label{P_memory_1}
    P^{p(j)}_t-(\bar{A}_{i^j_{t+1}}+\bar{B}_{i^j_{t+1}} K^{p(j)}_t)^T P^j_{t+1} (\bar{A}_{i^j_{t+1}}+\bar{B}_{i^j_{t+1}} K^{p(j)}_t)>0 
\end{equation}
$\forall (j,t+1)\in I_{\llbracket 1,N-1 \rrbracket}$, and 
\begin{equation}\label{P_memory_2}
     P^{p(j)}_{N-1}- (\bar{A}_{i^j_{N}}+\bar{B}_{i^j_{N}} K^{p(j)}_{N-1})^T P_0 (\bar{A}_{i^j_{N}}+\bar{B}_{i^j_{N}} K^{p(j)}_{N-1})> 0 
\end{equation}

\end{lemma}
\begin{proof}
See Appendix \ref{Appendix_proofs}. 
\end{proof}
 
The following theorem defines necessary and sufficient LMIs for the feedback gains $K^j_t$, $\forall (j,t)\in I_{\llbracket 0,N-1\rrbracket}$ of an LITPC.
\begin{theorem}\label{suff_theorem_differnt_gains}
    The closed-loop $\mathcal{S}$ is robustly exponentially stabilizable by an LITPC if and only if there exist $N \in \mathbb{Z}_{\geq 1}$, matrices $L^j_t\in \mathbb{R}^{n_u \times n_x}$, $(j,t)\in I_{\llbracket 0,N-1 \rrbracket}$ and symmetric matrices $S^j_t >0$, $(j,t)\in I_{\llbracket 0,N-1 \rrbracket}$, where $S^1_0=S_0$ such that 
\begin{equation}\label{LMI_1_diff_K}
\begin{pmatrix}
S^{p(j)}_t &  S^{p(j)^T}_t \bar{A}^T_{i^j_{t+1}}+L^{p(j)^T}_t \bar{B}^T_{i^j_{t+1}}\\
\bar{A}_{i^j_{t+1}} S^{p(j)}_t +\bar{B}_{i^j_{t+1}} L^{p(j)}_t   &  S^j_{t+1}
\end{pmatrix}> 0
\end{equation}
$\forall (j,t+1)\in I_{\llbracket 1,N-1 \rrbracket}$,   
\begin{equation}\label{LMI_2_diff_K}
\begin{pmatrix}
S^{p(j)}_{N-1} &   S^{p(j)^T}_{N-1} \bar{A}^T_{i^j_{N}}+L^{p(j)^T}_{N-1} \bar{B}^T_{i^j_{N}}\\
\bar{A}_{i^j_{N}} S^{p(j)}_{N-1} +\bar{B}_{i^j_{N}} L^{p(j)}_{N-1} &    S_0
\end{pmatrix}> 0
\end{equation}
$\forall (j,N)\in I_{N}$. A stabilizing LITPC of period $N$ can then be defined using the gains $K^j_t$ which are computed according to
\begin{equation}\label{K_eqn_diff_K}
    K^j_t=L^j_t S^{j^{-1}}_t, \textbf{ } \forall (j,t)\in I_{\llbracket 0,N-1 \rrbracket}.
\end{equation} 
\end{theorem}
\begin{proof}
See Appendix \ref{Appendix_proofs}. 
\end{proof}
\begin{remark}
In order to determine a stabilizing controller in a simple algorithmic fashion, one can increase $N$ until the set of LMIs \eqref{LMI_1_diff_K}, \eqref{LMI_2_diff_K} are feasible. Note that the period of the controller is the same as the period of the FSLF in the LMIs \eqref{LMI_1_diff_K}, \eqref{LMI_2_diff_K}, and thanks to the results provided in Appendix \ref{section_periodic_stab} (Corollary \ref{corollary_necc_suff_imp} which resulted from Theorem \ref{suff_theorem_fixed_K}), such controller will always be found if the system is at all stabilizable by a LITPC. 
This may not result in a controller of the shortest possible period. However, due to Theorem \ref{suff_theorem_fixed_K} and Corollary \ref{corollary_necc_suff_imp} (see Appendix \ref{section_periodic_stab}), it is guaranteed that such an algorithm will terminate with the shortest possible period of the FSLF and hence the least number of LMIs to be solved offline. 
\end{remark}
\begin{Example}\label{Example_2}
Consider the system \eqref{system_dynamics} with $A_t=\begin{pmatrix}
a^1_t & 1.5 \\ 1.22 & a^2_t
\end{pmatrix}$ and $B_t=\begin{pmatrix}
1\\1
\end{pmatrix}$, where $a^1_t\in [0.2,1.8]$, $a^2_t \in [-0.2,0.2]$. The quadratic stabilizability criterion, the method in \cite{canon_inv_2} and our method from section \ref{subsection_static_A} for finding a static feedback controller fail to find a stabilizing controller.  On the other hand, solving the LMIs in Theorem \ref{suff_theorem_differnt_gains} with $N=2$, a stabilizing periodic controller is found. The stabilizing controller has one gain matrix for $[t \mod 2]=0$ which is $K_0=\begin{pmatrix}
-1.3517 & 0.3123
\end{pmatrix}$, and four gain matrices at $[t \mod 2]=1$ which  are $K^1_1=\begin{pmatrix}
-1.3745  &  0.9952
\end{pmatrix}$, $K^2_1=\begin{pmatrix}
-1.2372  &  0.1555
\end{pmatrix}$, $K^3_1=\begin{pmatrix}
-1.3534  &  0.8483
\end{pmatrix}$, $K^4_1=\begin{pmatrix}
 -1.2460  &  0.1952
\end{pmatrix}$.
\end{Example}
\begin{remark}\label{remark_comparison}
Unlike the static controller that we proposed in section \ref{subsection_static_A} which may or may not produce less conservative results than the method in \cite{canon_inv_2}, by construction, the proposed LITPC synthesis method will always produce less conservative results than the method in \cite{canon_inv_2}. 
\end{remark}
In fact the method in \cite{canon_inv_2} can be considered as a conservative special case of our method in which the matrices $K^j_t$ and $P^j_t$ are equal $\forall j \in \{1,2,\dots,n^t_d \}$ at each $t \in \{ 0,1,\dots, N-1\}$, i.e., for each $t \in \{ 0,1,\dots, N-1\}$, $K_t=K^j_t$, $P_t=P^j_t$ $\forall j \in \{ 1,2,\dots, n^t_d\}$. 
\begin{remark}\label{remark_no_LMI_LITPC}
The number of LMIs that characterize the LITPC is $\sum^N_{k=1}n_d^k$. The number of matrices to be determined is $2\sum^{N-1}_{k=0}n_d^k$ which are $\sum^{N-1}_{k=0}n_d^k$ matrices $S^j_k$ (which determine the Lyapunov matrices $P^j_k$), and $\sum^{N-1}_{k=0}n_d^k$ matrices $L^j_k$ (which together with $S^j_k$ determine the gain matrices $K^j_k$). 
\end{remark}
\section{Constrained Systems and Periodic Invariance}\label{section_constrained}
We now extend our results from section \ref{Section_controller_synthesis} to constrained systems. Assume that the system \eqref{system_dynamics} is subject to the constraints
\begin{equation}\label{Constraints_eqn}
    Fx_t+Eu_t \leq \textbf{1} , \textbf{ } \forall t \in \mathbb{Z}_{\geq 0},
\end{equation}
where $F \in \mathbb{R}^{n_c \times n_x}$, $E \in \mathbb{R}^{n_c \times n_u}$ and $ \textbf{1}$ is a column vector of dimension $n_c \times 1$ with all its elements equal to $1$. Similar to the unconstrained case, consider that the inputs $u_t$ are generated by some periodic control law of period $N$. Again we will consider only the first $N$ time steps of the closed-loop system as they will determine the closed-loop properties $\forall t \in \mathbb{Z}_{\geq 0}$. We want to find a control law that stabilizes the closed-loop uncertain system while satisfying the imposed constraints \eqref{Constraints_eqn}. 
The following definition (adapted from \cite{canon_inv_2}) is a relaxation of positive invariance.
\begin{definition}\label{Definition_FPI}
A set $\mathcal{X}_0$ is called a Robust Periodic Invariant set with period $N \in \mathbb{Z}_{\geq 1}$ for the constrained closed-loop system $\mathcal{S}$, if $\forall x_0 \in \mathcal{X}_0$, the closed-loop states and inputs $x_t$ and $u_t$ satisfy \eqref{Constraints_eqn}, $\forall t \in \{0,1,\dots, N-1 \}$, $\forall \textbf{d}^N \in \mathbf{D}^N$, and $x_{N} \in \mathcal{X}_0$.    
\end{definition}
\subsection{Periodic Invariance Using a Static Feedback Gain}
In this section, we want to control the constrained system by a static feedback gain $K$. The closed-loop system dynamics and constraints are then
\begin{equation}\label{dynamics_constrained_1}
   x_{t+1}=(A_t+B_t K)x_t,
\end{equation}
\begin{equation}\label{dynamics_constrained_2}
   x_t \in \mathbb{X}= \{x | \textbf{ } (F+EK)x \leq \textbf{1} \}, \textbf{ } \forall t \in \mathbb{Z}_{\geq 0}.
\end{equation}
Define the ellipsoidal sets $\mathcal{P}^j_t=\{x| \textbf{ }x^T P^j_t x \leq 1\}$, $(j,t) \in I_{\llbracket 0,N-1 \rrbracket}$, where $P_0=P^1_0$ and $P^j_t=S^{j^{-1}}_t$. 

The following theorem provides convex conditions that suffice for an ellipsoidal set to be robust periodic invariant as well as for robust exponential stability of the closed-loop system.
\begin{theorem} \label{Theorem_FPI_Static_K}
    If there exists $N \in \mathbb{Z}_{\geq 1}$, a matrix $G \in \mathbb{R}^{n_x \times n_x}$, a matrix $L\in \mathbb{R}^{n_u \times n_x}$ and symmetric positive definite matrices $S^j_t \in \mathbb{R}^{n_x \times n_x}$, $(j,t)\in I_{\llbracket 0,N-1 \rrbracket}$, where $S^1_0=S_0$ such that \eqref{LMI_1_static} and \eqref{LMI_2_static} hold, and symmetric matrices $H^j_t \in \mathbb{R}^{n_c \times n_c}$, $(j,t)\in I_{\llbracket 0,N-1 \rrbracket}$ such that 
\begin{equation}\label{LMI_constraints_static_1}
\begin{pmatrix}
H^j_t &   & FG+EL \\
G^T F^T+L^T E^T &  & G^T+G-S^j_t 
\end{pmatrix}\geq 0
\end{equation}
$\forall (j,t)\in I_{\llbracket 0,N-1 \rrbracket}$, where for each $H^j_t$, $(j,t)\in I_{\llbracket 0,N-1 \rrbracket}$  
\begin{equation}\label{LMI_constraints_static_2}
e_i^T H^j_t e_i \leq 1 \textbf{ } \textbf{ } \forall i \in \{1,2,\dots,n_c \},
\end{equation}
where $e_i$ is column $i$ of a $n_c \times n_c$ identity matrix, then the ellipsoidal set $\mathcal{P}_0= \{ x| \textbf{ } x^T P_0 x \leq 1 \}$ where $P_0=S^{-1}_0$ is robust periodic invariant with period $N$ for the system $\mathcal{S}$, with the static feedback gain $K=L G^{-1}$. Moreover, if $x_0 \in \mathcal{P}_0$, then the static gain $K=L G^{-1}$ results in the robust exponential stability and constraint satisfaction of the closed-loop system $\mathcal{S}$.
\end{theorem}
\begin{proof}
    See Appendix \ref{Appendix_proofs}.
\end{proof}
Based on Theorem \ref{Theorem_FPI_Static_K}, we can solve the following convex program to maximize the volume of the robust periodic invariant set $\mathcal{P}_0$ (which can be achieved by minimizing the convex function $-log(det(S_0))$ \cite{Boyd:94})
\begin{align}\label{optimization_K}
 \underset{G, L, H^j_t, S^j_t \forall (j,t)\in I_{\llbracket 0,N-1\rrbracket} }{\text{minimize         }-log(det(S_0))}   & &\text{subject to: }\eqref{LMI_1_static}, \eqref{LMI_2_static}  , \eqref{LMI_constraints_static_1}, \eqref{LMI_constraints_static_2}. 
\end{align}

Note that among all the ellipsoidal sets $\mathcal{P}^j_t$, $\forall (j,t) \in I_{\llbracket 0,N-1\rrbracket}$, only the ellipsoidal set $\mathcal{P}_0$ is robust periodic invariant. Nevertheless, the following result can be stated.
Define the sets 
\begin{equation*}
\bar{\mathbb{P}}_t=Co(\{ \mathcal{P}^j_t \}, \textbf{ }j \in \{1,2,\dots,n_d^{t} \}),
\end{equation*}
$\forall t \in \{0,1,\dots,N-1 \}$.
\begin{corollary}\label{FPI_corollary_static}
The sets $\bar{\mathbb{P}}_t$, $t \in \{0,1,\dots,N-1 \}$ are robust periodic invariant for the closed-loop $\mathcal{S}$ controlled by a static controller determined as per Theorem \ref{Theorem_FPI_Static_K}.
\end{corollary}
\begin{proof}
    See Appendix \ref{Appendix_proofs}.
\end{proof}
 \subsection{Periodic Invariant Sets Using LITPC}\label{section_FPI_memory}
 In this section, we propose conditions for obtaining robust periodic invariant sets using a LITPC (see Definition \ref{Def_Memory_controller} and Algorithm \ref{algorithm_controller_periodic})  from section \ref{subsection_memory}. Considering the first $N$ time steps, the control law is $\kappa_{N}(t,\textbf{x}^{t}_{0})=\sum^{n_d^t}_{j=1}\beta^j_{t} K^j_t x^j_t$, where $x_t=\sum^{n_d^t}_{j=1}\beta^j_t x^j_t$, $t \in \{0,1,\dots,N-1 \}$. 

Define the ellipsoidal sets $\mathcal{P}^j_t=\{x| \textbf{ }x^T P^j_t x \leq 1\}$, $(j,t) \in I_{\llbracket 0,N-1 \rrbracket}$, where $P_0=P^1_0$ and $P^j_t=S^{j^{-1}}_t$. 

Define $ \mathbb{X}^j_t= \{x | \textbf{ } (F+EK^j_t)x \leq \textbf{1} \}$, $\forall (j,t)\in I_{\llbracket 0,N-1 \rrbracket}$.

It will be shown that if the closed-loop system $\mathcal{S}_\mathcal{D}$ satisfies the mixed state and input constraints, then the closed-loop system $\mathcal{S}$ satisfies the mixed state and input constraints.
\begin{theorem} \label{Theorem_FPI_periodic_K}
    If there exists $N \in \mathbb{Z}_{\geq 1}$, matrices $L^j_t\in \mathbb{R}^{n_u \times n_x}$, $(j,t)\in I_{\llbracket 0,N-1 \rrbracket}$, and symmetric positive definite matrices $S^j_t \in \mathbb{R}^{n_x \times n_x}$, $(j,t)\in I_{\llbracket 0,N-1 \rrbracket}$, where $S^1_0=S_0$ such that \eqref{LMI_1_diff_K} and \eqref{LMI_2_diff_K} hold, and symmetric matrices $H^j_t \in \mathbb{R}^{n_c \times n_c}$ such that 
\begin{equation}\label{LMI_constraints_periodic_1}
\begin{pmatrix}
H^j_t &   & F S^j_t+E L^j_t\\
 S^{j^T}_t F^T+ L^{j^T}_t E^T &  &  S^j_t
\end{pmatrix}\geq 0 
\end{equation}
$\forall (j,t)\in I_{\llbracket 0,N-1 \rrbracket}$,  where for each $H^j_t$, $(j,t)\in I_{\llbracket 0,N-1 \rrbracket}$    
\begin{equation}\label{LMI_constraints_periodic_2}
e_i^T H^j_t e_i \leq 1 \textbf{ } \textbf{ } \forall i \in \{1,2,\dots,n_c \},
\end{equation}
where $e_i$ is column $i$ of a $n_c \times n_c$ identity matrix, then the ellipsoidal set $\mathcal{P}_0= \{ x| \textbf{ } x^T P_0 x \leq 1 \}$ where $P_0=S^{-1}_0$ is robust periodic invariant with period $N$ for the system $\mathcal{S}$ when using the LITPC of period $N$, which has the gains $K^j_t=L_t^j S^{j^{-1}}_t$. Moreover, if $x_0 \in \mathcal{P}_0$ then for the closed-loop system $\mathcal{S}$, the LITPC results in robust exponential stability and constraint satisfaction for $\mathcal{S}$.
\end{theorem}
\begin{proof}
    See Appendix \ref{Appendix_proofs}.
\end{proof}

We can therefore solve the following convex program to maximize the volume of the robust periodic invariant set $\mathcal{P}_0$:
\begin{align}\label{optimization_differnt_K}
 \underset{L^j_t, H^j_t, S^j_t \forall (j,t)\in I_{\llbracket 0,N-1\rrbracket} }{\text{minimize         }-log(det(S_0))}   & &\text{subject to: }\eqref{LMI_1_diff_K}, \eqref{LMI_2_diff_K}  , \eqref{LMI_constraints_periodic_1}, \eqref{LMI_constraints_periodic_2}. 
\end{align}

Define the sets 
\begin{equation*}
\bar{\mathbb{P}}_t=Co(\{ \mathcal{P}^j_t \}, \textbf{ }j \in \{1,2,\dots,n_d^{t} \}),
\end{equation*}
$\forall t \in \{0,1,\dots,N-1 \}$.
Similar to the case of static controllers the following result holds. 
\begin{corollary}\label{Corollary_FPI_diff_K}
The sets $\bar{\mathbb{P}}_t$, $t \in \{0,1,\dots,N-1 \}$ are robust periodic invariant for the closed-loop $\mathcal{S}$ controlled by a LITPC for which the gains are specified in Theorem \ref{Theorem_FPI_periodic_K}.
\end{corollary}
\begin{proof}
    See Appendix \ref{Appendix_proofs}.
\end{proof}
The robust periodic invariant sets in \cite{canon_inv_2} are explained in Figure \ref{Fig_ill_per_inv}. On the other hand, the robust periodic invariant sets proposed in this section are illustrated in Figure \ref{new_Fig_ill_per_inv}.
\begin{figure}[!h]
	\begin{center}
		\includegraphics[width=1\columnwidth]{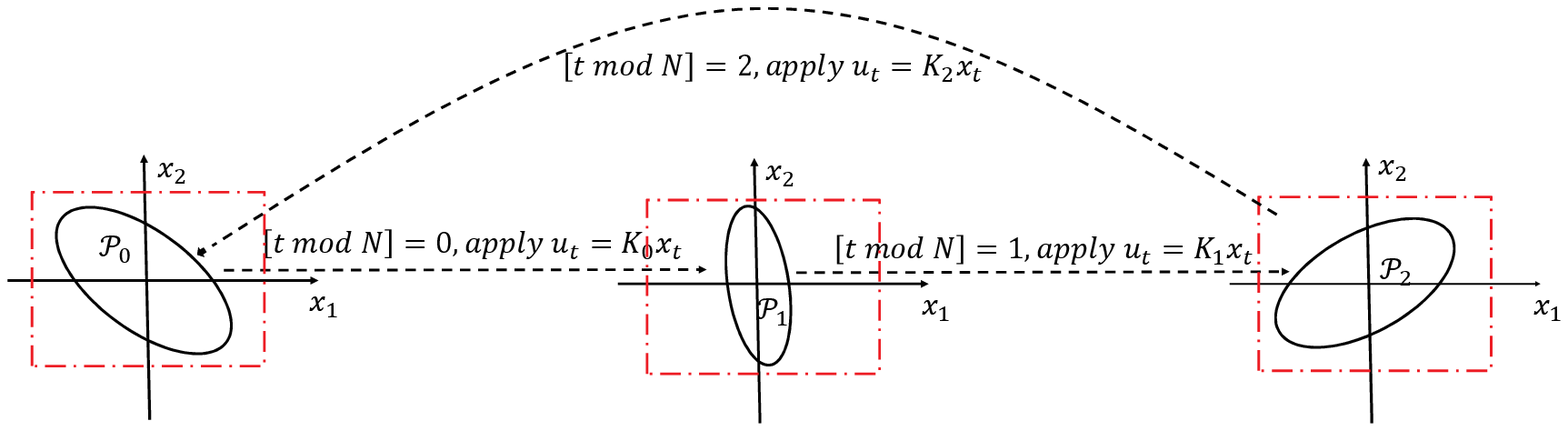} 
		\caption{The figure shows a sketch of the ellipsoidal robust periodic invariance that results from the approach in \cite{canon_inv_1,canon_inv_2}, with $N=3$. The constraints are shown as dashed-dotted red boxes. If the state is inside the ellipsoid $\mathcal{P}_0$, applying $u_t=K_0x_t$  guarantees that the next state is inside $\mathcal{P}_1$, $\forall (A_t,B_t)\in \mathbf{D}$. At the next time step, $u_t=K_1x_t$ is applied to the plant which guarantees that the state at the next time step will be inside $\mathcal{P}_2$. At the next time step $u_t=K_2x_t$ moves the state back again to $\mathcal{P}_0$.}
		\label{Fig_ill_per_inv}  
	\end{center}
\end{figure} 
\begin{figure}[!h]
	\begin{center}
		\includegraphics[width=1\columnwidth]{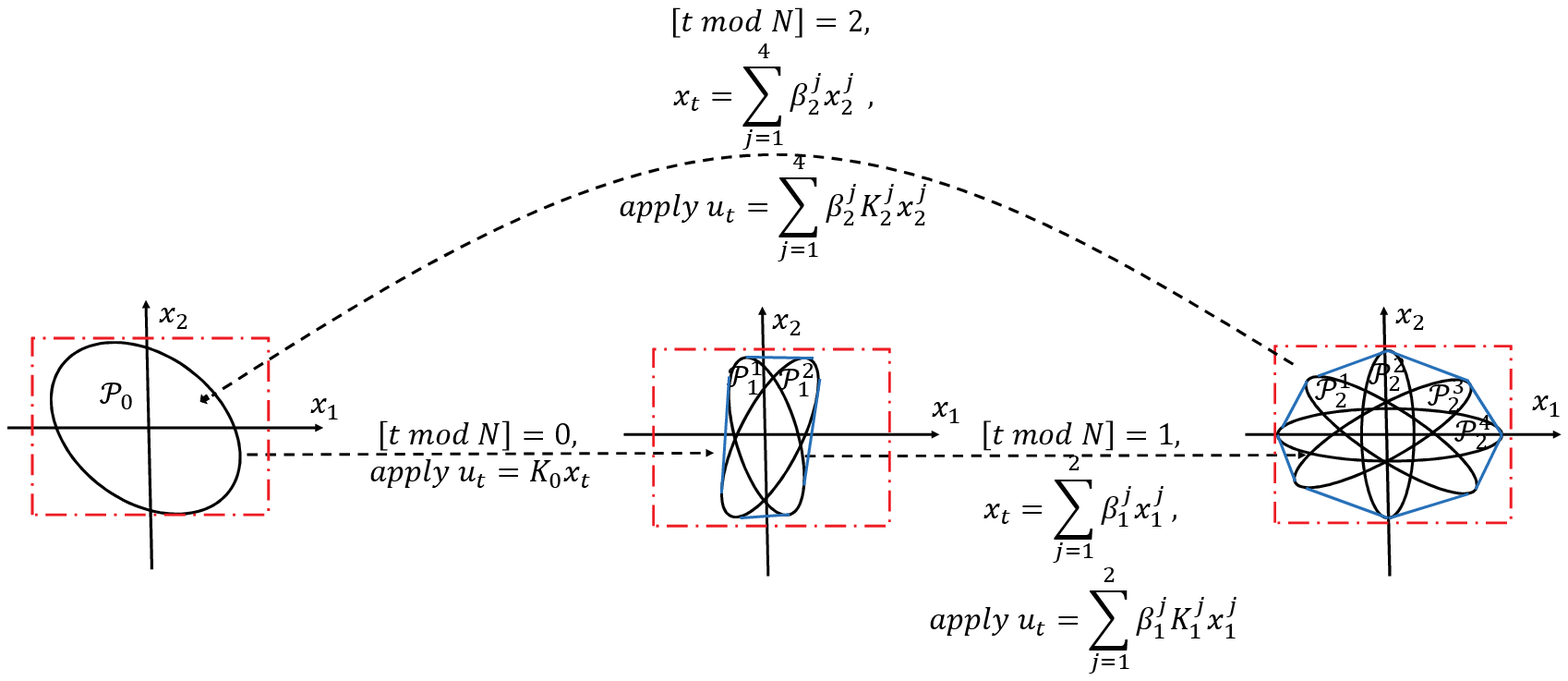} 
		\caption{The figure shows a sketch of the ellipsoidal robust periodic invariance resulting that results from the application of the LITPC with $N=3$ for an uncertain system with $n_d=2$. The constraints are shown by dashed-dotted red boxes. If the state is inside the ellipsoid $\mathcal{P}_0$, applying $u_t=K_0x_t$  guarantees that the next state is inside the convex hull of $\mathcal{P}^1_1$ and $\mathcal{P}^2_1$, $\forall (A_t,B_t)\in \mathbf{D}$. At the next time step, determining vectors $x^j_1 \in \mathcal{P}^j_1$ and scalars $\beta^j_1\geq0$ such that $\sum^2_{j=1}\beta^j_1=1$ and $x_t=\sum^2_{j=1}\beta^j_1x^j_1$, and applying $u_t=\sum^2_{j=1}\beta^j_1K^j_1x^j_1$ to the plant guarantees that the state at the next time step will be inside the convex hull of $\mathcal{P}^j_2$, for $j=\{1,2,3,4\}$. At the next time step, determining vectors $x^j_2 \in \mathcal{P}^j_2$ and scalars $\beta^j_2\geq0$ such that $\sum^4_{j=1}\beta^j_2=1$ and $x_t=\sum^4_{j=1}\beta^j_2x^j_2$, and applying $u_t=\sum^4_{j=1}\beta^j_2K^j_2x^j_2$ to the plant guarantees that the state at the next time step will be inside $\mathcal{P}_0$.}
		\label{new_Fig_ill_per_inv}  
	\end{center}
\end{figure}
\begin{Example}\label{example_sets}
Consider the system \eqref{system_dynamics} with $A_t=\begin{pmatrix}
a^1_t & 1.2 \\ a^2_t & 0.1
\end{pmatrix}$ and $B_t=\begin{pmatrix}
1\\1
\end{pmatrix}$, where $a^1_t\in [0.2,1.75]$, $a^2_t \in [-0.6,0.6]$. The state and input constraints are $\begin{pmatrix}
-1.5\\-1.5
\end{pmatrix}\leq x\leq \begin{pmatrix}
1.5\\1.5
\end{pmatrix}$, $-1\leq u \leq 1$.
We compare the size of the robust periodic invariant set $\mathcal{P}_0$ that results from the newly proposed static and periodic controllers with the size of the set that is obtained by using the controller from \cite{canon_inv_2}. The set $\mathcal{P}_0$ from each of these methods for several values of $N$ is shown in Figure \ref{Sets_P_0_all}. 
Finding a robust periodic invariant set using a static controller by solving \eqref{optimization_K} with $N=2$ results in $4.6 \%$ increase in the volume of $\mathcal{P}_0$ over the value that results from \cite{canon_inv_2} with $N=20$, while using $N=4$ results in an increase of $7.7 \%$ in the volume of $\mathcal{P}_0$ over the result from \cite{canon_inv_2} with $N=20$. Note that this increase in volume is achieved using an appealing (from a complexity point of view) static controller. Using the proposed LITPC synthesis to find a robust periodic invariant set by solving \eqref{optimization_differnt_K} with $N=2$ results in an increase of $12.5 \%$ in the volume of $\mathcal{P}_0$ over the result from \cite{canon_inv_2} with $N=20$, while using $N=4$ results in an increase of $30.5 \%$ in the volume of $\mathcal{P}_0$ over the result from \cite{canon_inv_2} with $N=20$. 

Figure \ref{Sets_P_0_1_diff_K} shows the sets $\mathcal{P}^j_t$, $\forall (j,t) \in I_{\llbracket 0,N-1\rrbracket}$ for the periodic controller that results from solving \eqref{optimization_differnt_K} with $N=2$. Since the uncertainty set $\mathcal{D}$ has four elements, we have four different ellipsoidal sets $\mathcal{P}^j_1$ at $t=1$. Note that the ellipsoidal sets $\mathcal{P}^j_1$,  $j \in \{ 1,2,3,4\}$ are not contained in the set $\mathcal{P}_0$ (which is the reason of the reduced conservatism). The set of states that can be reached from $\mathcal{P}_0$, $\forall (A,B)\in \mathbf{D}$ is a subset of the convex hull of these four ellipsoids. At $t=2$, the state is guaranteed to be inside $\mathcal{P}_0$. 
Figure \ref{Sets_P_0_1_diff_K_4} shows the resulting sets for $N=4$. Again the ellipsoidal sets $\mathcal{P}^j_t$,  $(j,t) \in I_{\llbracket 1,3\rrbracket}$ are not contained in the set $\mathcal{P}_0$.

\begin{figure}[!h]
\begin{center}
\includegraphics[width=\columnwidth]{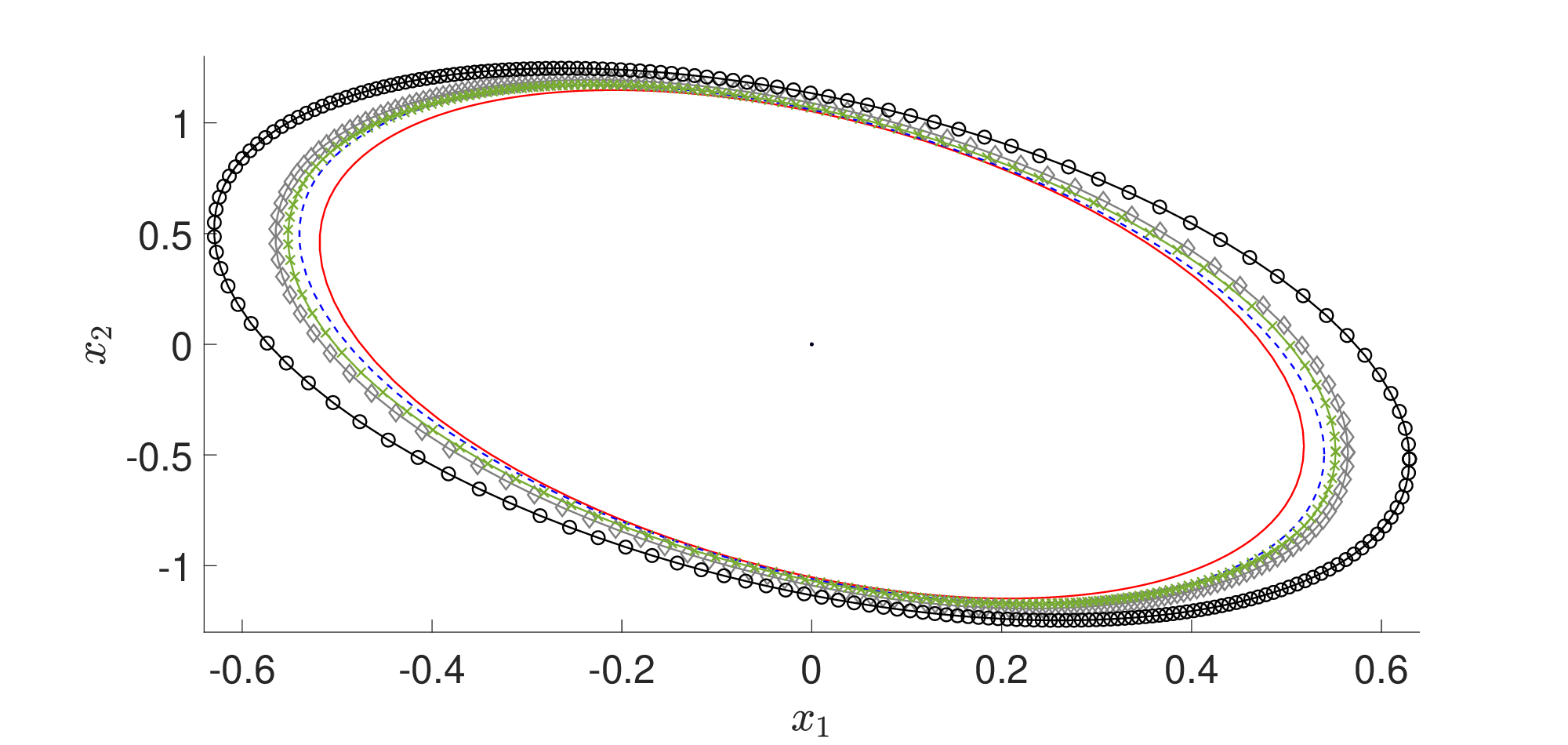} 
\caption{Example \ref{example_sets}: The sets $\mathcal{P}_0$. The set obtained by the method from \cite{canon_inv_2} with $N=20$ (solid red). Sets obtained by the proposed static controller with $N=2$ (dashed blue) and $N=4$ (cross Green). Sets obtained by the proposed LITPC with $N=2$ (diamond grey) and $N=4$ (circle black). }
\label{Sets_P_0_all}  
\end{center} 
\end{figure}
\begin{figure}[!h]
\begin{center}
\includegraphics[width=\columnwidth]{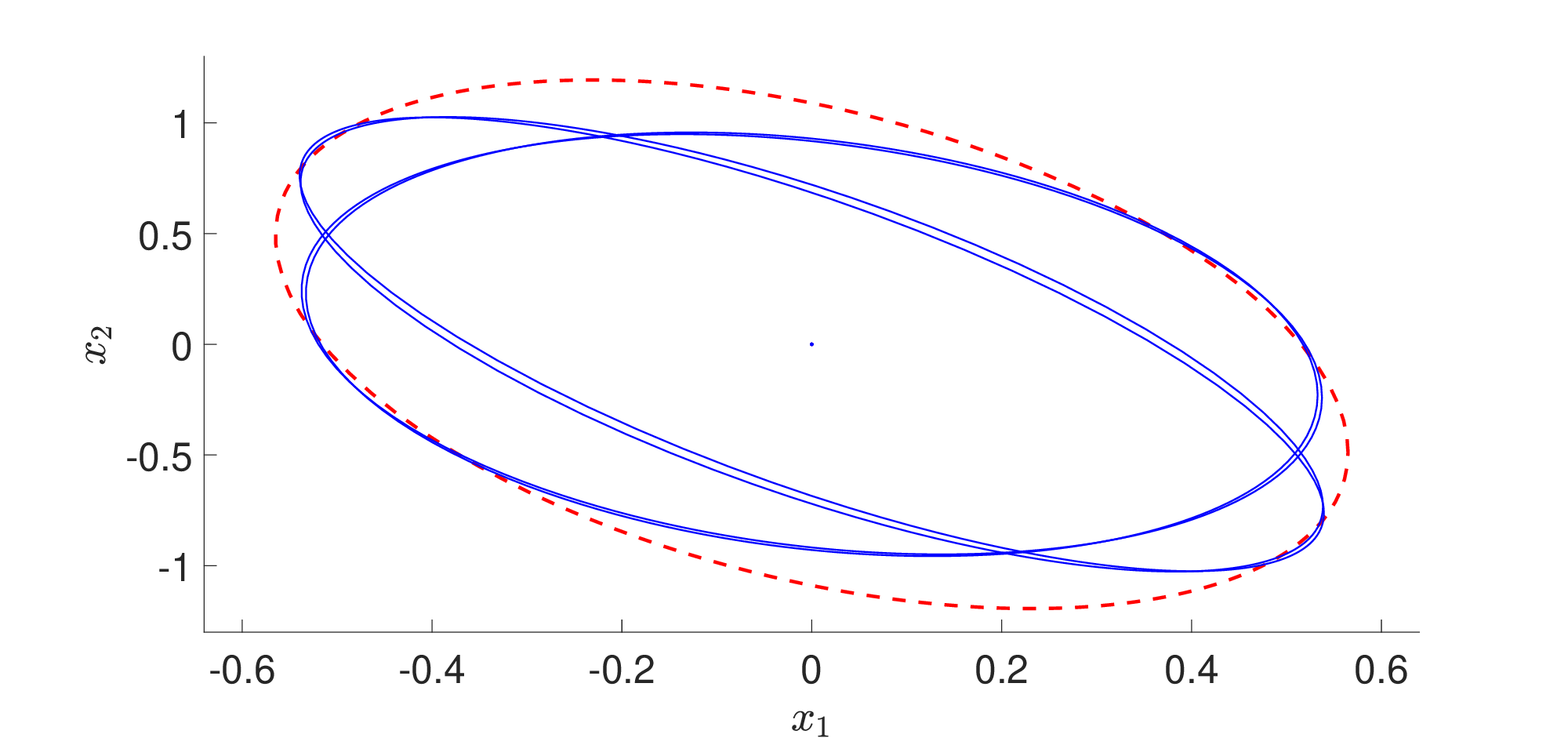}
\caption{Example \ref{example_sets}: The sets $\mathcal{P}^j_t$ resulting from solving \eqref{optimization_differnt_K} with $N=2$. The set $\mathcal{P}_0$ is shown in dashed red. The sets $\mathcal{P}^j_1$, $j \in \{ 1,2,3,4\}$ are shown in solid blue.}
\label{Sets_P_0_1_diff_K}  
\end{center} 
\end{figure}
\begin{figure}[!h]
\begin{center}
\includegraphics[width=\columnwidth]{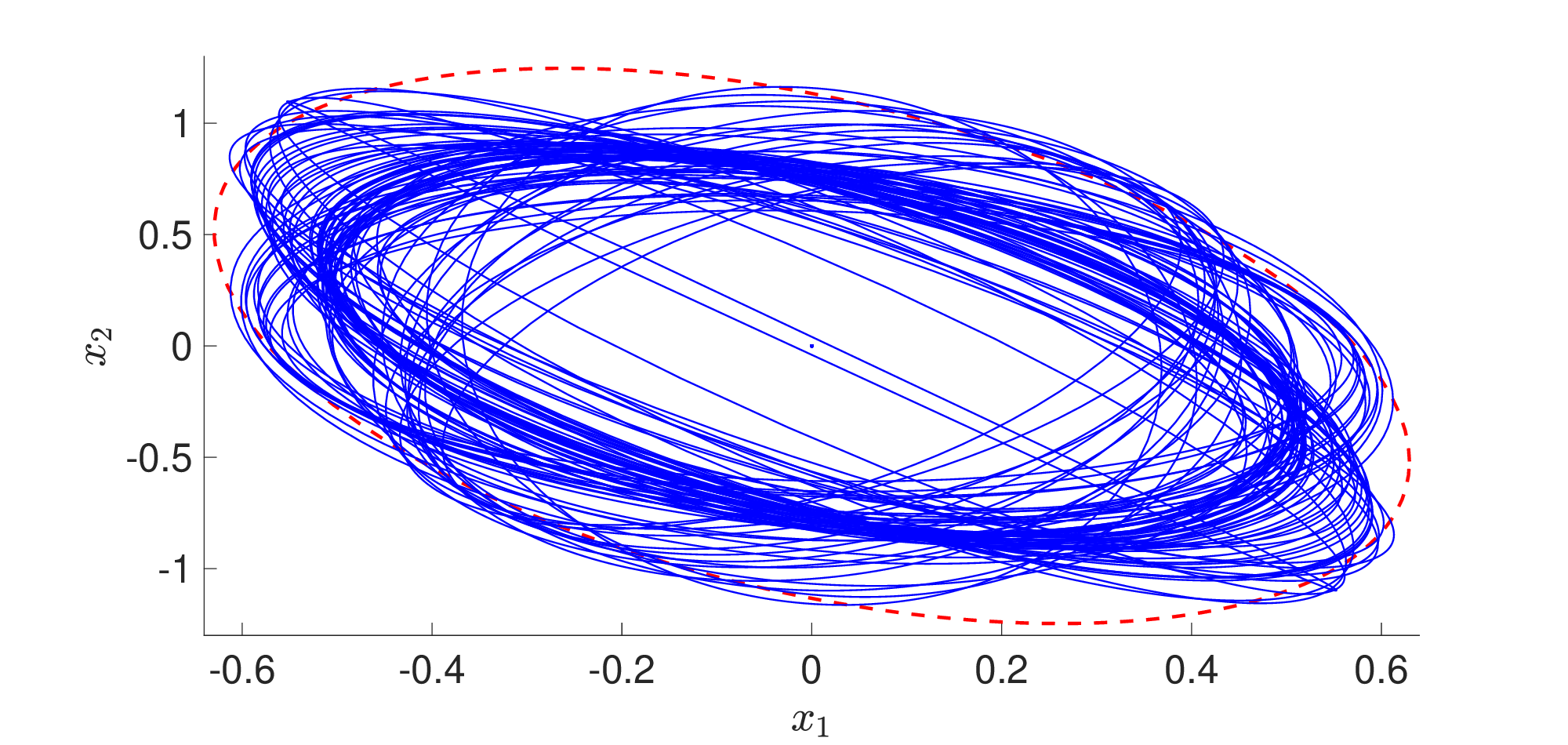}
\caption{Example \ref{example_sets}: The sets $\mathcal{P}^j_t$ resulting from solving \eqref{optimization_differnt_K} with $N=4$. The set $\mathcal{P}_0$ is shown in dashed red. The sets $\mathcal{P}^j_t$, $(j,t) \in I_{\llbracket 1,3 \rrbracket}$ are shown in solid blue.}
\label{Sets_P_0_1_diff_K_4}  
\end{center} 
\end{figure}
\end{Example}

\begin{remark}

Similar to what we have mentioned for the unconstrained case (Remark \ref{remark_comparison}), the method in \cite{canon_inv_2} may or may not provide a larger robust periodic invariant set for the same $N$ than the one that results from solving \eqref{optimization_K}. However, for the same $N$ solving \eqref{optimization_differnt_K} will always provide a robust periodic invariant set which is bigger than or equal to the one from the method in \cite{canon_inv_2}. This comes at the expense of the exponential increase in the number of LMIs in the optimization problem \eqref{optimization_differnt_K} that need to be solved offline. 
\end{remark}
\begin{remark}\label{remark_MPC_large}
    The main interest in maximizing the set $\mathcal{P}_0$ in the offline optimization \eqref{optimization_differnt_K} (or \eqref{optimization_K}) is that this set can be used as the terminal set for many existing periodic MPC formulations (i.e., MPC with cyclic prediction horizons) such as \cite{KOGEL2013809,LAZAR_MPC_periodic}. Note that for these MPC formulation the sets $\mathcal{P}^j_t$, $\forall (j,t) \in \in I_{\llbracket 1,N-1 \rrbracket}$ (i.e., the ellipsoids in solid blue in Figures \ref{Sets_P_0_1_diff_K}, \ref{Sets_P_0_1_diff_K_4}) are not needed for the MPC formulation and only $\mathcal{P}_0$ (the ellipsoid in dashed red in Figures \ref{Sets_P_0_1_diff_K}, \ref{Sets_P_0_1_diff_K_4}) is needed.
\end{remark}

\section{Conclusion}
New necessary and sufficient conditions for robust stabilization of uncertain linear systems by periodic controllers were derived by employing finite step Lyapunov functions, and novel convex criteria for obtaining robust stabilizing controllers and ellipsoidal periodic invariant sets for constrained linear systems were proposed by utilizing scenario trees along with quadratic finite step Lyapunov functions. For unconstrained systems, less conservative static controllers were obtained compared to other methods, as well as non-conservative linear interpolating tree periodic controllers. We extended these results to constrained systems and derived convex offline criteria for obtaining periodic invariant ellipsoids using both static and linear interpolating tree periodic controllers. It was demonstrated by numerical examples that the conservatism that results from our approach is reduced in comparison with existing methods. We expect that our findings will have an impact on the design of new robust MPC schemes, which will be addressed in our future work. 

\printbibliography 

\appendix
\section{Characterization for Stabilization Using Quadratic Finite Step Lyapunov Functions}\label{section_periodic_stab}
 
Assume that the system \eqref{system_dynamics} is controlled by a periodic controller of period $N \in \mathbb{Z}_{\geq 1}$. At time step $t$, the sequence of past and current states of length $k+1$ where $k=[t \mod N] \in \{0,1,\dots, N-1 \}$, which starts at the beginning of the period and ends at the current time step $t=mN+k$ where $m=[t/N]$, is denoted by 
\begin{equation*}
    \textbf{x}^{t}_{t-k}=\textbf{x}^{mN+k}_{mN}=\{ x_{mN}, x_{mN+1},\dots, x_{mN+k} \}.
\end{equation*}
For example, if $´N=3$, then 
\begin{align*}
    \text{At }t=0, &\textbf{ } \textbf{x}^{0}_{0}=\{ x_{0} \}, \\
    \text{At }t=1, &\textbf{ } \textbf{x}^{1}_{0}=\{ x_{0}, x_1 \},\\
    \text{At }t=2, &\textbf{ } \textbf{x}^{2}_{0}=\{ x_{0}, x_1,x_2 \},\\
    \text{At }t=3, &\textbf{ } \textbf{x}^{3}_{3}=\{ x_{3} \},\\
    \text{At }t=4, &\textbf{ } \textbf{x}^{4}_{3}=\{ x_3, x_4 \},\\
    \text{At }t=5, &\textbf{ } \textbf{x}^{5}_{3}=\{ x_3, x_4,x_5 \},\\
    \vdots
\end{align*}

Consider a periodic control law $\kappa_N(k,\textbf{x}^{mN+k}_{mN})$, that operates on the present as well as the past $k$ states, where $N$ is the period of the controller, $t=mN+k$ and $k=[t \mod N] \in \{0,1,\dots, N-1 \}$. The controller itself will be denoted by $\kappa_N$. For the purpose of this section, assume that this control law results in an uncertain linear closed-loop system which has periodic uncertainty sets as follows
\begin{equation}
    x_t= x_{k+mN}= \Phi_{k+mN-1}x_{mN}, \label{linear_eqn_wierd}
\end{equation}
$\forall m \in \mathbb{Z}_{\geq 0}$, $k \in \{1,2,\dots,N \}$, where $\Phi_{k+mN} \in \mathbf{D}_k^{\kappa_N} \subset \mathbb{R}^{n_x \times n_x}$, and $\mathbf{D}_k^{\kappa_N}$, $\forall k \in \{0,\dots,N-1 \}$ are compact sets. 

For example with $N=3$, we have 
\begin{align*}
    x_1=\Phi_0x_0, \text{ } & \Phi_0 \in \mathbf{D}_0^{\kappa_3}, \\
    x_2=\Phi_1x_0, \text{ } & \Phi_1 \in \mathbf{D}_1^{\kappa_3} , \\
    x_3=\Phi_2x_0, \text{ }& \Phi_2 \in \mathbf{D}_2^{\kappa_3}, \\
     x_4=\Phi_3x_3,  \text{ } &\Phi_3 \in \mathbf{D}_0^{\kappa_3}, \\
    x_5=\Phi_4x_3,  \text{ } &\Phi_4 \in \mathbf{D}_1^{\kappa_3} , \\
    x_6=\Phi_5x_3,  \text{ } &\Phi_5 \in \mathbf{D}_2^{\kappa_3},\\
    \vdots
\end{align*}

At $t=k+mN$, the matrices $\Phi_t=\Phi_{k+mN}$ depend on the realized sequence of $(A_t,B_t) \in \mathbf{D}$ in interval $t \in \{mN,mN+1, \dots, mN+k \}$ as well as the sequence of control laws in that same interval. Hence, such periodic nature of the uncertainty set of the closed-loop (the sets $\mathbf{D}_k^{\kappa_N}$) is only due to the periodic nature of the control law. The connection between \eqref{system_dynamics} and \eqref{linear_eqn_wierd} is explained in Remark \ref{remark_Phi_Imp} and the text above it. Indeed, the N-step system, i.e., $x_{(m+1)N}=\Phi_{(m+1)N-1}x_{mN}$, $\Phi_{(m+1)N-1} \in \mathbf{D}^{\kappa_N}_{N-1}$ is uncertain and linear in the traditional sense.
Let $\bar{\kappa}_N$ denote the sequence of $N$ control laws in the period, i.e.,

\begin{equation}\label{seq_control_laws}
    \bar{\kappa}_N=\{ \kappa_N(0,\cdot),\kappa_N(1,\cdot),\dots,\kappa_N(N-1,\cdot) \}.
\end{equation} 

For the stability analysis of linear discrete time uncertain systems, it was shown in \cite{Megretski,pier_switches_FSLFS}, that the existence of a FSLF is necessary and sufficient for the stability of the linear difference inclusion. Using FSLFs, we derive necessary and sufficient conditions for the stabilization of uncertain linear systems by periodic controllers in the following theorem. 

\begin{theorem}\label{suff_theorem_fixed_K}
Consider the closed-loop system \eqref{linear_eqn_wierd}.
	\begin{itemize}
	    \item [A.]  Assume a periodic controller of period $N$. If there exists a FSLF with period $N$ on $\mathbb{R}^{n_x}$ for the resulting closed-loop system, then the closed-loop system \eqref{linear_eqn_wierd} is robustly exponentially stable on $\mathbb{R}^{n_x}$.
	    \item[B.] Assume that the system is robustly exponentially stabilizable on $\mathbb{R}^{n_x}$ by a periodic controller. For any function $V$ of the form \eqref{global_FSLF_eqn}, there exists a periodic controller of period $N$ which robustly exponentially stabilizes the system on $\mathbb{R}^{n_x}$, for which $V$ is a FSLF with the same period $N$.   
	\end{itemize}
	
\end{theorem}

\begin{proof}

 \textbf{ }
    \begin{itemize}
    \item [A. ] The function $V$ is a quadratic Lyapunov function for the $N$-step system, therefore, the $N$-step system is robustly exponentially stable. Hence, there exists $c_0\geq1$, and $\rho \in (0,1)$ such that for any $t=mN$, $\forall m \in \mathbb{Z}_{\geq 0}$, 
    \begin{equation}
        \|x_{mN}\| \leq \quad c_0 \rho^m \| x_0\|, \label{e_3}
    \end{equation}
       $\forall x_0 \in \mathbb{R}^{n_x}$, $\forall \textbf{d}^t \in \mathbf{D}^t$. Let $\mathbb{T}=\{0,1,\dots,N-1 \}$. Let $\vartheta=max \left( 1, \underset{\Phi \in \underset{k \in \mathbb{T}}{\cup}\mathbf{D}^{\kappa_N}_k}{max}\|\Phi\| \right)$. Therefore, $\forall t=m N+k$, $m \in \mathbb{Z}_{\geq 0}$ and $\forall \textbf{d}^t \in \mathbf{D}^t$,
     \begin{equation}\label{eq_stab_final}
         \|x_t\| \overset{\eqref{linear_eqn_wierd}}{\leq}  \vartheta \| x_{mN}\|  \overset{\eqref{e_3}}{\leq}  c_0 \vartheta \rho^{\frac{t-k}{N}} \| x_0\| \leq c \lambda^{t} \| x_0\|, 
     \end{equation}
     where $c= c_0 \vartheta \rho^{-\frac{N-1}{N}} $ and $\lambda=\rho^{\frac{1}{N}}$. Clearly $c\geq 1$ and $\lambda \in (0,1)$ which completes the proof of part A.
     \item[B.] Since $P_0>0$, 
    \begin{equation}\label{upper_lower_bounds}
        c_l \| x \|^2 \leq V(x)\leq c_u \| x \|^2,
    \end{equation}
    $\forall x\in \mathbb{R}^{n_x}$, where $c_l>0$ and $c_u>0$ are the minimum and maximum eigenvalues of $P_0$. Since the system is robustly exponentially stabilizable by a periodic controller, then there exists a periodic controller $\kappa_{N_c}$ of some period $N_c$, for which the closed-loop is robustly exponentially stable. The sequence of control laws in the period is defined in \eqref{seq_control_laws}. Consider any function $V$ of the form \eqref{global_FSLF_eqn}, with $P_0>0$. Due to robust exponential stability of the closed-loop by that controller, there exists constants $c \geq 1$ and $\lambda \in (0,1)$, such that, $\forall x_0 \in \mathbb{R}^{n_x}$, $\forall \textbf{d}^t\in \mathbf{D}^t$, $\|x_t\| \leq  c \lambda^t \| x_0 \|$. Therefore, 
    \begin{equation}
        V(x_t) \leq  c^2 (\lambda^2)^t \frac{c_u}{c_l}V(x_0), \label{to_be_rep_2}
    \end{equation}
    Hence, $\forall x_0 \in \mathbb{R}^{n_x} \setminus \{ 0\}$, we have $V(x_{M_s})<V(x_0)$, $\forall \textbf{d}^{M_s}\in \mathbf{D}^{M_s}$, for any positive integer $M_s>  \frac{\log \left(\frac{c_l}{c_u c^2} \right)}{2 \log(\lambda)} $. Let $N_s$ be the smallest integer greater than $\frac{\log \left(\frac{c_l}{c_u c^2} \right)}{2 \log(\lambda)}$.
    Note that $N_s$ is not necessarily less than or equal to $N_c$ and hence we have the following two cases.
    
    \textbf{Case 1 ($N_s\leq N_c$):} In that case, $V(x_{N_c})<V(x_0)$, $\forall x_0 \in \mathbb{R}^{n_x} \setminus \{0\}$, $\forall \textbf{d}^{N_c}\in \mathbf{D}^{N_c}$, is equivalent to $ \Phi^T P_0\Phi< P_0$, $\forall \Phi \in \mathbf{D}^{\kappa_{N_c}}_{N_c-1}$, which implies $V(x_{(m+1)N_c}) < V(x_{mN_c})$, $\forall m \in \mathbb{Z}_{\geq 0}$, $\forall x_{mN_c} \neq 0$ which completes the proof with $N=N_c$.
        
    \textbf{Case 2 ($N_s>N_c$):} Consider a periodic controller of period $N_s$ that we call $\kappa_{N_s}$ which is obtained by the periodic repetition of the sequence $\bar{\kappa}_{N_c}$ of the controller $\kappa_{N_c}$, 
        \begin{equation}
            \bar{\kappa}_{N_s}=\{ \textbf{ } \overbrace{\bar{\kappa}_{N_c}, \bar{\kappa}_{N_c}, \dots, \bar{\kappa}_{N_c}}^{[\sfrac{N_s}{N_c}] \text{ times}}, \kappa_{N_c}(0,\cdot), \kappa_{N_c}(1,\cdot), \dots, \kappa_{N_c}(j,\cdot) \}  ,
        \end{equation}
    where $j=[N_s \mod N_c]-1$. 
    The closed-loop system \eqref{linear_eqn_wierd} that results from using the controller $\kappa_{N_s}$ is then defined using the sets $\mathbf{D}^{\kappa_{N_s}}_{l}$, $\forall l \in \{0,1,\dots, N_s-1 \}$  which are defined using some recursive multiplication of the elements of the sets $\mathbf{D}^{\kappa_{N_c}}_0, \mathbf{D}^{\kappa_{N_c}}_{1},\dots, \mathbf{D}^{\kappa_{N_c}}_{{N_c}-1}$.

        $\forall m \in \{0,1,\dots, \left[ \sfrac{N_s}{N_c} \right] \}$, $k \in \{0,1,\dots,N_c-1 \}$ with $mN_c+k \in \{0, 1, \dots,  N_s-1\}$. 
        Since $[\sfrac{N_s}{N_c}]$ is finite and the sets $\mathbf{D}^{\kappa_{N_c}}_k$, $k \in \{0,1,\dots, N_c-1 \}$ are compact, then sets $\mathbf{D}^{\kappa_{N_s}}_l$, $\forall l \in \{0,1,\dots, N_s-1 \}$ are compact. Therefore, using $\kappa_{N_s}$, we have $ \Phi^TP_0\Phi< P_0$, $\forall \Phi \in \mathbf{D}^{\kappa_{N_s}}_{N_s-1}$, which is equivalent to $V(x_{(m+1)N_s}) < V(x_{mN_s})$, $\forall m \in \mathbb{Z}_{\geq 0}$, $\forall x_{mN_s} \neq 0$, $\forall \textbf{d}^{(m+1)N_s}\in \mathbf{D}^{(m+1)N_s}$. Since $V(x_{(m+1)N_s}) < V(x_{mN_s})$, $\forall m \in \mathbb{Z}_{\geq 0}$, $\forall x_{mN_s} \neq 0$, $\forall \textbf{d}^{(m+1)N_s}\in \mathbf{D}^{(m+1)N_s}$, robust exponential stability of the closed-loop by the controller $\kappa_{N_s}$ is preserved due to part A of the Theorem, which completes the proof with $N=N_s$.
\end{itemize}
\end{proof}
A direct consequence of Theorem \ref{suff_theorem_fixed_K} is the following result which will be important for controller synthesis.
\begin{corollary}\label{corollary_necc_suff_imp}
The system \eqref{linear_eqn_wierd} is robustly exponentially stabilizable by periodic state feedback if and only if there exists a periodic control law of period $N$ and a quadratic FSLF of period $N$ for the resulting closed-loop \eqref{linear_eqn_wierd}, where $N \in \mathbb{Z}_{\geq 1}$.
\end{corollary}
A well know result (see \cite{MOLCHANOV198959,Megretski,Lee_stab_LPV} for example) is that asymptotic stability (AS) of compact linear differential and difference inclusions is equivalent to exponential stability (ES). For the sake of completeness, we formally show that this result also holds for systems of the form \eqref{linear_eqn_wierd}.
\begin{corollary}
For the closed-loop system \eqref{linear_eqn_wierd}, asymptotic stability is equivalent to exponential stability.
\end{corollary}
\begin{proof}
    Exponential stability of \eqref{linear_eqn_wierd} implies asymptotic stability by definition. It remains to show the converse. If \eqref{linear_eqn_wierd} is asymptotically stable, then the N-step version of the system, i.e., $x_{(m+1)N}=\Phi_{(m+1)N-1}x_{mN}$, $\Phi_{(m+1)N-1} \in \mathbf{D}^{\kappa_N}_{N-1}$ is asymptotically stable, and hence exponentially stable (because it is a standard compact linear difference inclusion). As was shown in the proof of Theorem \ref{suff_theorem_fixed_K} (See \eqref{e_3} and \eqref{eq_stab_final}), ES of the N-step system implies ES of the original system defined in \eqref{linear_eqn_wierd}, which completes the proof. 
\end{proof}

\section{Proofs}\label{Appendix_proofs}
\textbf{Proof of Lemma \ref{convexity_lemma}:}
\begin{proof}
    The proof of part A has a similarity to proofs from the MPC literature that exploit the vertices of the uncertainty sets (see \cite{Scokaert_Mayne},\cite{dela_pena} for example). For the system $\mathcal{S}$, $(A_t,B_t)=\sum^{n_d}_{i=1}\alpha_{t,i} (\bar{A}_i,\bar{B}_i)$ (see Assumption \ref{Assumption_1}). Using \eqref{system_dynamics} and Assumption \ref{Assumption_1}, at $t=0$, $x_1=\sum^{n_d}_{j=1}\alpha_{0,j}(\bar{A}_j x_0+\bar{B}_j u_0)=\sum^{n_d}_{j=1}\alpha_{0,j}x^j_{1}=\sum^{n_d}_{j=1}\beta^j_1 x^j_{1}$. Applying this recursively and noting how the input is defined in step I2 in Algorithm \ref{algorithm_controller_periodic}, we have
\begin{align}
x_{t+1} & =  A_t x_t+B_t u_t  =\sum^{n_d}_{i=1}\alpha_{t,i}(\bar{A}_i x_t+\bar{B}_i u_t) \\
  & =\sum^{n_d}_{i=1}\alpha_{t,i} (\sum^{n_d^t}_{p(j)=1}\beta^{p(j)}_{t}(\bar{A}_i x^{p(j)}_t+\bar{B}_i  u^{p(j)}_t))\\
  & =\sum^{n^{t+1}_d}_{j=1}\beta^j_{t+1}x^j_{t+1}, \textbf{ } \forall t \in \{0,1,\dots,N-1 \}\label{convex_state}
\end{align}
which proves part A.
For part B, note that the possible trajectories of $\mathcal{S}_{\mathcal{D}}$ are a subset of the possible trajectories of $\mathcal{S}$, which implies that a FSLF for $\mathcal{S}$ is a FSLF for $\mathcal{S}_{\mathcal{D}}$. It remains to prove the converse. At $t=N$, \eqref{convex_state} gives $x_N=\sum^{n^{N}_d}_{j=1}\beta^j_{N}x^j_{N}$, for any $x_0 \in \mathbb{R}^{n_x}$, for any $ \textbf{d}^N \in \mathbf{D}^N$. For $\mathcal{S}_\mathcal{D}$ we have $x^T_0 P_0 x_0 -x^{j^T}_N P_0 x^j_N>0$, $\forall (j,N) \in I_N$, $\forall x_0 \in \mathbb{R}^{n_x} \setminus \{ 0\}$ which by the Schur complement is equivalent to
\begin{equation}\label{equation_shur_iteself}
    \begin{pmatrix}
x^T_0 P_0 x_0  &  x^{j^T}_N\\
x^{j}_N &  P^{-1}_0
\end{pmatrix}> 0, \textbf{ } \forall (j,N)\in I_N.
\end{equation}
Multiplying \eqref{equation_shur_iteself} by $\beta^j_{N}$ for any $\beta^j_{N}$ that satisfy Definition \ref{beta_definition}, and taking the sum over all $j\in \{ 1,2,\dots,n_d^N\}$, we get 
\begin{equation}\label{equation_shur_iteself_2}
    \begin{pmatrix}
x^T_0 P_0 x_0  &  x^{T}_N\\
x_N &  P^{-1}_0
\end{pmatrix}> 0,
\end{equation}
which by the Schur complement is equivalent to 
\begin{equation*}
    x^T_0 P_0 x_0 -x^{T}_N P_0 x_N>0, \textbf{ }\forall x_0 \in \mathbb{R}^{n_x}\setminus \{ 0\}, \textbf{ }\forall \textbf{d}^N \in \mathbf{D}^N.
\end{equation*}
\end{proof}
\textbf{Proof of Lemma \ref{prop_suff_FSLF_K_fixed}:}
\begin{proof}
See also Figure \ref{Tree_first} as a visual support for the proof. Consider first the system $\mathcal{S}_\mathcal{D}$. Using \eqref{P_static_1} recursively along the tree until $t=N-2$ and multiplying the result from the left and the right by any non-zero $x^T_0$ and $x_0$ gives: 
\begin{equation}\label{intermediate_P_static_2}
    x^T_0  P_0  x_0 > x^{j^T}_{N-1}  P^{j}_{N-1}  x^{j}_{N-1}, \forall (j,N-1)\in I_{N-1}.
\end{equation}
If $x^{p(j)}_{N-1} \neq 0$, $(p(j),t) \in  I_{N-1}$, \eqref{P_static_2} is equivalent to $x^{p(j)^T}_{N-1}  P^{p(j)}_{N-1}  x^{p(j)}_{N-1}> x^{j^T}_{N}  P_0  x^{j}_{N}, \textbf{ } \forall (j,N)\in I_{N}$, which when used with \eqref{intermediate_P_static_2} gives $x^T_0  P_0  x_0>x^{j^T}_{N}  P_0  x^{j}_{N}$, $ \forall (j,N)\in I_{N}$. Therefore, the function $V(x)=x^TP_0x$ is a FSLF for $\mathcal{S}_\mathcal{D}$.

Since the inputs of the scenario tree are computed by $u^j_t=Kx^j_t$, and $\sum^{n_d^t}_{j=1}\beta^j_{t} u^j_t=K \sum^{n_d^t}_{j=1}\beta^j_{t} x^j_t$ where $x_t=\sum^{n_d^t}_{j=1}\beta^j_{t} x^j_t$, the control law $\kappa(x_t)=Kx_t$ satisfies Algorithm \ref{algorithm_controller_periodic}. Therefore, due to Lemma \ref{convexity_lemma}.B, the function $V(x)=x^TP_0x$ is a FSLF for the system $\mathcal{S}$.

Let $\bar{A}^K_i=\bar{A}_i+\bar{B}_iK$, $\forall i \in \Gamma$. Define $\mathcal{D}_t^{\kappa_N}=\mathcal{D}^K=\{ \bar{A}^K_i, \textbf{ }i \in \Gamma \}$, and $\mathbf{D}_t^{\kappa_N}=\mathbf{D}^K=Co(\mathcal{D}^K)$ (Compare with section \ref{section_periodic_stab}). Note that the time-invariance of the closed-loop uncertainty set $\mathbf{D}_t^{\kappa_N}$ is due to having a linear time invariant control law. The system $\mathcal{S}$ evolves according to $x_{t+1}=A^K_tx_t$, where $A^K_t\in \mathbf{D}^K$, $\forall t \in \mathbb{Z}_{\geq 0}$.

Robust exponential stability of the closed-loop then holds from \cite{Megretski} because the set $\mathbf{D}^K$ is constant $\forall t \in \mathbb{Z}_{\geq 0}$.  
\end{proof}

\textbf{Proof of Theorem \ref{Theorem_stability_static_suff}:}
\begin{proof}
    Let $G=S^{j}_t+\epsilon^{j}_t$, for some $\epsilon^{j}_t\in \mathbb{R}^{n_x \times n_x}$, $\forall (j,t)\in I_{\llbracket 0,N-1 \rrbracket}$, $\epsilon^1_0=\epsilon_0$. Also define $P^j_t=S^{j^{-1}}_t$. Therefore,
    \begin{equation}\label{G_eqn_1}
        G^T P^j_t G= S^j_t +\epsilon^{j^T}_t P^j_t \epsilon^j_t + \epsilon^{j^T}_t +\epsilon^j_t, 
    \end{equation}
  $\forall (j,t)\in I_{\llbracket 0,N-1 \rrbracket}$ . Since  $P^j_t>0$, $\forall (j,t)\in I_{\llbracket 0,N-1 \rrbracket}$, $\epsilon^{j^T}_t P^j_t \epsilon^j_t\geq 0$. Using this with \eqref{G_eqn_1} and the definition of $\epsilon^{j}_t$ we get
 \begin{equation}\label{G_eqn_3}
      G^T P^j_t G\geq G^T + G- S^j_t. 
  \end{equation}
If \eqref{LMI_1_static} and \eqref{LMI_2_static} hold, then $G^T+G-S^j_t>0$, hence $G$ is of full rank. Using the Schur complement of \eqref{LMI_1_static}, and using the bound derived in \eqref{G_eqn_3}, we get
\begin{equation}\label{shur_eqn_1}
    G^T P^{p(j)}_t G - (\bar{A}_{i^j_{t+1}} G+\bar{B}_{i^j_{t+1}}L)^T P^{j}_{t+1} (\bar{A}_{i^j_{t+1}} G+\bar{B}_{i^j_{t+1}}L)> 0.
\end{equation}
Since $G$ is of full rank, we can let $L=KG$ for some $K\in \mathbb{R}^{n_u \times n_x}$, and then \eqref{shur_eqn_1} is equivalent to
\begin{equation}\label{shur_eqn_2}
     P^{p(j)}_t - (\bar{A}_{i^j_{t+1}}+\bar{B}_{i^j_{t+1}}K)^T P^{j}_{t+1} (\bar{A}_{i^j_{t+1}}+\bar{B}_{i^j_{t+1}}K)> 0,
\end{equation}
$\forall (j,t+1)\in I_{\llbracket 1,N-1\rrbracket}$. Using the same arguments for \eqref{LMI_2_static}, implies 
\begin{equation}\label{shur_eqn_last}
     P^{p(j)}_{N-1} - (\bar{A}_{i^j_{N}}+\bar{B}_{i^j_{N}}K)^T P_{0} (\bar{A}_{i^j_{N}}+\bar{B}_{i^j_{N}}K)> 0,
\end{equation}
$\forall (j,N)\in I_{N}$. Therefore, from \eqref{shur_eqn_2}, \eqref{shur_eqn_last} and  Lemma \ref{prop_suff_FSLF_K_fixed}, the control law $\kappa(x_t)$ with $K$ defined in \eqref{K_eqn} stabilizes $\mathcal{S}$. 
\end{proof}

\textbf{Proof of Lemma \ref{prop_suff_FSLF_K_different} :}
\begin{definition}\label{Definition_scenario}
For the scenario tree of length $N$, we will define a scenario $\mathbf{S}_j$ as the sequence $(\bar{A}_{i^{f(j)}_{t+1}},\bar{B}_{i^{f(j)}_{t+1}})$, $\forall t+1 \in \{ 1,2,\dots,N\}$ that starts at $x_0$ and ends at $x^j_N$, where $j\in \{1,2\dots,n_d^N \}$, and $(\bar{A}_{i^{f()}_{t+1}},\bar{B}_{i^{f(s)}_{t+1}})$ is the uncertainty realization at time step $t$ in scenario $\mathbf{S}_j$.  For example, $\mathbf{S}_2$ in Figure \ref{Tree_first} is the sequence $\{(\bar{A}_1,\bar{B}_1),\textbf{ } (\bar{A}_1,\bar{B}_1), \textbf{ } (\bar{A}_2,\bar{B}_2)\}$.
\end{definition}
Note that the definition of a scenario in this work is different from the definition of a scenario in \cite{SALA_scenario}. 

Define $ \psi^j_{0:1}=\bar{A}_{i^j_{1}}+\bar{B}_{i^j_{1}} K_0$, and
\begin{equation}\label{Psi_equation_for_d}
    \psi^j_{0:t+1}=(\bar{A}_{i^j_{t+1}}+\bar{B}_{i^j_{t+1}} K^{p(j)}_t) \psi^{p(j)}_{0:t}, \textbf{ } \forall (j,t+1) \in I_{\llbracket 1,N \rrbracket}.
\end{equation}

Define $\mathcal{D}^{\kappa_N}_{t}=\{ \psi^j_{0:t+1}, \textbf{ } \forall j \in \{1,2,\dots,n^{t+1}_d \})\}$, $\forall t \in \{0,\dots,N-1 \}$. Note that for any $t \in \{1,2,\dots, N \}$, $x_t=\sum^{n_d^t}_{j=1}\beta^j_{t} x^j_t=\sum^{n_d^t}_{j=1} \beta^j_{t} \psi^j_{0:t} x_0$, $\forall t\in \{ 1,2, \dots, N \}$, and hence (compare with section \ref{section_periodic_stab})
\begin{equation}\label{eqn_x_t_Phi}
    x_t=\Phi_{t-1}x_0, \textbf{ } \forall t\in \{ 1,2, \dots, N \},
\end{equation}
 where $\Phi_{t}\in \mathbf{D}^{\kappa_N}_{t}\subseteq Co(\mathcal{D}^{\kappa_N}_{t})$,  $\forall t \in \{0,1,\dots,N-1 \}$, where by construction $\mathbf{D}^{\kappa_N}_{t}$ are compact $\forall t \in \{0,1,\dots,N-1 \}$.
\begin{remark}\label{remark_Phi_Imp}
    Note that for each $t \in \{0,1,\dots,N-1 \}$ a different set of gains $K^j_t$, $j \in \{1, 2,\dots, n^t_d \}$ is used by the controller. As a result, unlike in case of the static controller in section \ref{subsection_static_A} (see the proof of Lemma \ref{prop_suff_FSLF_K_fixed}), the closed-loop uncertainty sets $\mathbf{D}^{\kappa_N}_t$ are not constant with $t$, and hence the results from \cite{Megretski} do not directly apply. This is one of the reasons for proving Theorem \ref{suff_theorem_fixed_K}.\\
\end{remark}
Note that the existence of a symmetric positive definite matrix $P_0$, such that $\forall x_0 \in \mathbb{R}^{n_x} \setminus \{ 0\}$, $x^{j^T}_N P_0 x^j_N<x^T_0 P_0 x_0$, $\forall (j,N)\in I_{N}$ for some $N \in \mathbb{Z}_{\geq 1}$ is equivalent to:
\begin{equation}\label{inequality_main_periodic}
    P_0- \psi^{j^T}_{0:N} P_0 \psi^j_{0:N}>0, \textbf{ } \forall (j,N)\in I_{N}.
\end{equation}

\begin{proof}

We consider first the system $\mathcal{S}_\mathcal{D}$ and prove the equivalence between \eqref{inequality_main_periodic} and the existence of $P^j_t$ such that \eqref{P_memory_1}, \eqref{P_memory_2} hold. We then use Lemma \ref{convexity_lemma}.B to extend the result to $\mathcal{S}$ and use Corollary \ref{corollary_necc_suff_imp} to complete the proof. 

Consider the system $\mathcal{S}_\mathcal{D}$. The proof of sufficiency (i.e., \eqref{P_memory_1}, \eqref{P_memory_2} $\implies$ \eqref{inequality_main_periodic}) follows the same steps as the proof of Lemma \ref{prop_suff_FSLF_K_fixed} for the case of a fixed gain $K$.

For the necessity, i.e, \eqref{inequality_main_periodic}$\implies$ existence of $P^j_t$ such that \eqref{P_memory_1}, \eqref{P_memory_2} hold (this part of the proof is illustrated below in a simplified manner in Illustration \ref{proof_illustration}), define for each scenario $\mathbf{S}_j$ (see Definition \ref{Definition_scenario}), $N-1$ matrices $M^j_t$, for $t \in \{1,2,\dots,N-1\}$ by backwards recursion as follows: 
\begin{equation}\label{M_matrices_definition}
    M^j_{t}=(\bar{A}_{i^{f(j)}_{t+1}}+\bar{B}_{i^{f(j)}_{t+1}} K^{p(f(j))}_{t})^T M^j_{t+1} (\bar{A}_{i^{f(j)}_{t+1}}+\bar{B}_{i^{f(j)}_{t+1}} K^{p(f(j))}_{t}),
\end{equation}
with the terminal condition $M^j_N=P_0$, $\forall j \in \{1,2,\dots, n_d^N \}$. Since $P_0>0$, therefore $M^j_t \geq 0$. Note that $\psi^{j^T}_{0:N} P_0 \psi^j_{0:N}=(\bar{A}_{i^{f(j)}_{1}}+\bar{B}_{i^{f(j)}_{1}} K^{p(f(j))}_{0})^T M^j_1 (\bar{A}_{i^{f(j)}_{1}}+\bar{B}_{i^{f(j)}_{1}} K^{p(f(j))}_{0})$, $\forall j \in \{1,2,\dots, n_d^N \}$, therefore the condition \eqref{inequality_main_periodic} can be re-written as
\begin{equation}\label{n_d^N_ineq}
    P_0-(\bar{A}_{i^{f(j)}_{1}}+\bar{B}_{i^{f(j)}_{1}} K^{p(f(j))}_{0})^T M^j_1 (\bar{A}_{i^{f(j)}_{1}}+\bar{B}_{i^{f(j)}_{1}} K^{p(f(j))}_{0})>0,
\end{equation}
where $K^{p(f(j))}_{0}=K_0$, $\forall j \in \{1,2,\dots, n_d^N \}$. Note that for each $n_d^{N-1}$ inequalities from \eqref{n_d^N_ineq}, where $j \in  \{ n^{N-1}_d (i-1)+1,\dots, n_d^{N-1}i \}$ and $i \in \Gamma$, $(\bar{A}_{i^{f(j)}_{1}}, \bar{B}_{i^{f(j)}_{1}})$ is the same, i.e., $(\bar{A}_{i^{f(j)}_{1}}, \bar{B}_{i^{f(j)}_{1}})=(\bar{A}_i,\bar{B}_i)$, for some $i \in \Gamma$. Therefore, for each $i \in \Gamma$
\begin{equation}\label{P_0_M_matrices}
   P_0-(\bar{A}_i+\bar{B}_iK_0)^T M^j_1 (\bar{A}_i+\bar{B}_iK_0)>0, 
\end{equation}
$\forall j \in \{ n^{N-1}_d (i-1)+1,\dots, n_d^{N-1}i\} $.  Therefore, applying Lemma \ref{essential_lemma} to \eqref{P_0_M_matrices}, we have that for each $i \in \Gamma$ there exists a matrix $P^i_1$ satisfying $P^i_1>M^j_1$, $j \in  \{ n^{N-1}_d (i-1)+1,\dots, n_d^{N-1}i \}$ such that 
\begin{equation}\label{P_0_M_matrices_2}
   P_0-(\bar{A}_i+\bar{B}_iK_0)^T P^i_1 (\bar{A}_i+\bar{B}_iK_0)>0. 
\end{equation}
Note that \eqref{P_0_M_matrices_2} is the same as \eqref{P_memory_1} for $t=0$. Since $P^i_1>M^j_1$, $ \forall j \in  \{ n^{N-1}_d (i-1)+1,\dots, n_d^{N-1}i \}$, $i \in \Gamma$ and from the definition of $M^j_t$ in \eqref{M_matrices_definition} we have for each $i \in \Gamma$, 
\begin{equation}
     P^i_1-(\bar{A}_{i^{f(j)}_{2}}+\bar{B}_{i^{f(j)}_{2}} K^{p(f(j))}_{1})^T M^j_2 (\bar{A}_{i^{f(j)}_{2}}+\bar{B}_{i^{f(j)}_{2}} K^{p(f(j))}_{1})>0,
\end{equation}
$\forall j \in  \{ n^{N-1}_d (i-1)+1,\dots, n_d^{N-1}i \}$, where $K^{p(f(j))}_1=K^i_1$, $i\in \Gamma$. Note that for each $i \in \Gamma$ this is the same as \eqref{n_d^N_ineq} but with $P^i_1$ instead of $P_0$ used at the root node of an $N-1$ step scenario tree instead of an N-step scenario tree, and with $M^j_2$ instead of $M^j_1$. Hence, applying the same reasoning that was used for $t=0$, \eqref{P_memory_1} holds for $t=1$. By induction until the end of the tree, and since $M^j_N=P_0$, $\forall j \in \{1,2,\dots, n_d^N \}$, \eqref{inequality_main_periodic} implies the existence of matrices $P^j_t$ such that both \eqref{P_memory_1} and \eqref{P_memory_2} hold which completes the proof.

Due to Lemma \ref{convexity_lemma}.B, the function $x^TP_0x$ is a FSLF for $\mathcal{S}$ if and only if it is FSLF for $\mathcal{S}_\mathcal{D}$. The result then holds by comparing \eqref{eqn_x_t_Phi} to \eqref{linear_eqn_wierd} and using Corollary \ref{corollary_necc_suff_imp} which guarantees the existence of such $N \in \mathbb{Z}_{\geq 1}$ if and only if the system is stabilizable by LITPC.

\begin{figure}
\begin{center}
\includegraphics[width=\columnwidth]{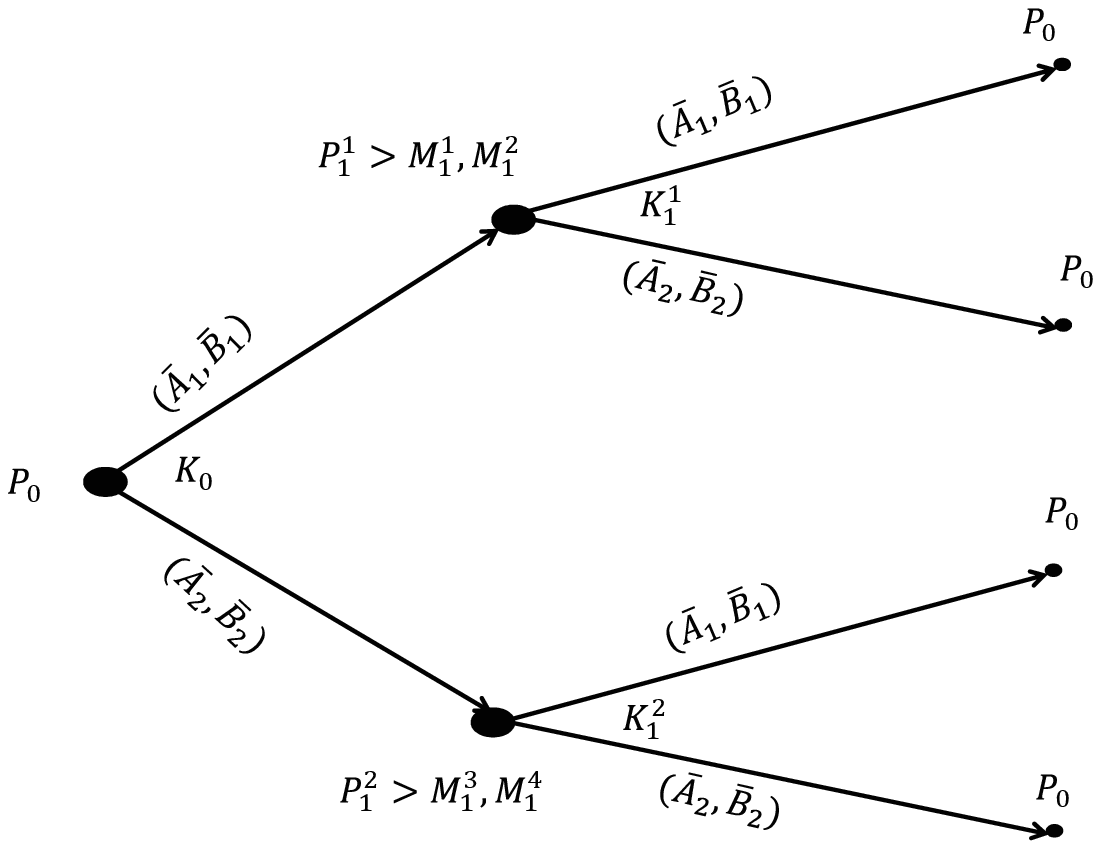} 
\caption{Illustration of the necessity part of the proof of Lemma \ref{prop_suff_FSLF_K_different}. The \textit{matrix replacement lemma} (Lemma \ref{essential_lemma}) is used to replace $M^1_1$, $M^2_1$ by $P^1_1$, and $M^3_1$, $M^4_1$ by $P^2_1$.}
\label{Proof_lemma}  
\end{center} 
\end{figure}
\end{proof}
\begin{Illustration}\label{proof_illustration}
The necessity part of the proof of Lemma \ref{prop_suff_FSLF_K_different} is illustrated by an example for which $N=2$, $\Gamma=\{1,2\}$ (see also  Figure \ref{Proof_lemma}). We want to prove that if \eqref{inequality_main_periodic} holds then there exists matrices $P^j_t$ such that \eqref{P_memory_1} and \eqref{P_memory_2} hold. For the considered case \eqref{inequality_main_periodic} is written as follows:
\begin{align}
    P_0>&(\Bar{A}_1+\Bar{B}_1K_0)^T(\Bar{A}_1+\Bar{B}_1K^1_1)^TP_0(\Bar{A}_1+\Bar{B}_1K^1_1)(\Bar{A}_1+\Bar{B}_1K_0),\label{eq_1^1}\\
   P_0>&(\Bar{A}_1+\Bar{B}_1K_0)^T(\Bar{A}_2+\Bar{B}_2K^1_1)^TP_0(\Bar{A}_2+\Bar{B}_2K^1_1)(\Bar{A}_1+\Bar{B}_1K_0),\label{eq_1^2}\\
    P_0>&(\Bar{A}_2+\Bar{B}_2K_0)^T(\Bar{A}_1+\Bar{B}_1K^2_1)^TP_0(\Bar{A}_1+\Bar{B}_1K^2_1)(\Bar{A}_2+\Bar{B}_2K_0), \label{eq_1^3}\\
   P_0>&(\Bar{A}_2+\Bar{B}_2K_0)^T(\Bar{A}_2+\Bar{B}_2K^2_1)^TP_0(\Bar{A}_2+\Bar{B}_2K^2_1)(\Bar{A}_2+\Bar{B}_2K_0). \label{eq_1^4}
\end{align}

Define
\begin{align}
    M^1_1=(\Bar{A}_1+\Bar{B}_1K^1_1)^TP_0(\Bar{A}_1+\Bar{B}_1K^1_1),\label{Meq_1^1}\\
    M^2_1=(\Bar{A}_2+\Bar{B}_2K^1_1)^TP_0(\Bar{A}_2+\Bar{B}_2K^1_1),\label{Meq_2^1}\\
     M^3_1=(\Bar{A}_1+\Bar{B}_1K^2_1)^TP_0(\Bar{A}_1+\Bar{B}_1K^2_1),\label{Meq_3^1}\\
    M^4_1=(\Bar{A}_2+\Bar{B}_2K^2_1)^TP_0(\Bar{A}_2+\Bar{B}_2K^2_1).\label{Meq_4^1}
\end{align}

We now use \eqref{Meq_1^1} in \eqref{eq_1^1}, \eqref{Meq_2^1} in \eqref{eq_1^2}, \eqref{Meq_3^1} in \eqref{eq_1^3} and \eqref{Meq_4^1} in \eqref{eq_1^4}. As a result \eqref{eq_1^1},\eqref{eq_1^2}, \eqref{eq_1^3}, \eqref{eq_1^4} are equivalent to 

\begin{align}
    P_0>&(\Bar{A}_1+\Bar{B}_1K_0)^TM^1_1(\Bar{A}_1+\Bar{B}_1K_0),\label{New_eq_1^1}\\
   P_0>&(\Bar{A}_1+\Bar{B}_1K_0)^TM^2_1(\Bar{A}_1+\Bar{B}_1K_0),\label{New_eq_1^2}\\
    P_0>&(\Bar{A}_2+\Bar{B}_2K_0)^TM^3_1(\Bar{A}_2+\Bar{B}_2K_0), \label{New_eq_1^3}\\
   P_0>&(\Bar{A}_2+\Bar{B}_2K_0)^TM^4_1(\Bar{A}_2+\Bar{B}_2K_0). \label{New_eq_1^4}
\end{align}
Applying Lemma \ref{essential_lemma}, to \eqref{New_eq_1^1} and \eqref{New_eq_1^2} we conclude that there exists a symmetric matrix $P^1_1$ that satisfies 
\begin{align}
P^1_1>M^1_1=(\Bar{A}_1+\Bar{B}_1K^1_1)^TP_0(\Bar{A}_1+\Bar{B}_1K^1_1),\label{Peq_1^1}\\
    P^1_1>M^2_1=(\Bar{A}_2+\Bar{B}_2K^1_1)^TP_0(\Bar{A}_2+\Bar{B}_2K^1_1),\label{Peq_2^1}
    \end{align}
and
\begin{equation}
    P_0>(\Bar{A}_1+\Bar{B}_1K_0)^TP^1_1(\Bar{A}_1+\Bar{B}_1K_0). \label{last_1}
\end{equation}

Similarly, applying Lemma \ref{essential_lemma}, to \eqref{New_eq_1^3} and \eqref{New_eq_1^4} we conclude that there exists a symmetric matrix $P^2_1$ that satisfies 
\begin{align}
P^2_1>M^3_1=(\Bar{A}_1+\Bar{B}_1K^2_1)^TP_0(\Bar{A}_1+\Bar{B}_1K^2_1),\label{Peq_3^1}\\
    P^2_1>M^4_1=(\Bar{A}_2+\Bar{B}_2K^2_1)^TP_0(\Bar{A}_2+\Bar{B}_2K^2_1),\label{Peq_4^1}
    \end{align}
and
\begin{equation}
    P_0>(\Bar{A}_2+\Bar{B}_2K_0)^TP^2_1(\Bar{A}_2+\Bar{B}_2K_0). \label{last_2}
\end{equation}

This means that if \eqref{eq_1^1}, \eqref{eq_1^2}, \eqref{eq_1^3}, \eqref{eq_1^4} hold then there exist $P^1_1>0$ and $P^2_1>0$ such that \eqref{Peq_1^1},\eqref{Peq_2^1}, \eqref{last_1}, \eqref{Peq_3^1}, \eqref{Peq_4^1}, \eqref{last_2} hold. 

Note that in the considered case \eqref{P_memory_1}, \eqref{P_memory_2} are \eqref{Peq_1^1},\eqref{Peq_2^1}, \eqref{last_1}, \eqref{Peq_3^1}, \eqref{Peq_4^1}, \eqref{last_2}.
\end{Illustration}

\textbf{Proof of Theorem \ref{suff_theorem_differnt_gains}:}
\begin{proof}
Define $P^j_t=S^{j^{-1}}_t$, $\forall (j,t)\in I_{\llbracket 0,N-1\rrbracket}$. Applying the Schur complement, \eqref{LMI_1_diff_K} is equivalent to
\begin{equation}\label{eq_S_diff_K_1}
      S^{p(j)}_{t} -  (\bar{A}_{i^j_{t+1}} S^{p(j)}_t + \bar{B}_{i^j_{t+1}} L^{p(j)}_t )^T P^j_{t+1} (\bar{A}_{i^j_{t+1}} S^{p(j)}_t + \bar{B}_{i^j_{t+1}} L^{p(j)}_t) \geq 0.
\end{equation}

Left and right multiplying \eqref{eq_S_diff_K_1} by $S^{p(j)^{-1}}_{t}$ and using \eqref{K_eqn_diff_K}, we deduce that \eqref{eq_S_diff_K_1} (and hence \eqref{LMI_1_diff_K}) is equivalent to
\begin{equation}\label{eq_S_diff_K_2}
    P^{p(j)}_{t} -  (\bar{A}_{i^j_{t+1}} + \bar{B}_{i^j_{t+1}} K^{p(j)}_t )^T P^j_{t+1}  (\bar{A}_{i^j_{t+1}} + \bar{B}_{i^j_{t+1}} K^{p(j)}_t  ) \geq 0.
\end{equation}
In the same fashion, \eqref{LMI_2_diff_K} is equivalent to 
\begin{equation}\label{eq_S_diff_K_3}
     P^{p(j)}_{N-1}- (\bar{A}_{i^j_{N}} + \bar{B}_{i^j_{N}} K^{p(j)}_{N-1})^T P_0( \bar{A}_{i^j_{N}} + \bar{B}_{i^j_{N}}K^{p(j)}_{N-1})> 0.
\end{equation}

The theorem then holds due to the equivalences \eqref{eq_S_diff_K_2}$ \iff$\eqref{LMI_1_diff_K}, \eqref{eq_S_diff_K_3}$\iff$\eqref{LMI_2_diff_K} and Lemma \ref{prop_suff_FSLF_K_different}. 
\end{proof}

\textbf{Proof of Theorem \ref{Theorem_FPI_Static_K}:}
\begin{proof}
From Theorem \ref{Theorem_stability_static_suff}, \eqref{LMI_1_static} and \eqref{LMI_2_static} imply robust exponential stability of the unconstrained closed-loop system. Therefore, it remains to show that \eqref{LMI_constraints_static_1} and \eqref{LMI_constraints_static_2} imply constraint satisfaction of the closed-loop $\mathcal{S}$ if $x_0 \in \mathcal{P}_0$, and the robust periodic invariance property of $\mathcal{P}_0$.

As shown in the proof of Theorem 2.9 in \cite{Kouvaritakis2016}, any ellipsoid ($\mathcal{P}^j_t$) is contained in a polytope ($\mathbb{X}$) defined by \eqref{dynamics_constrained_2}, if and only if there exists a matrix $H^j_t \in \mathbb{R}^{n_c \times n_c}$ for which \eqref{LMI_constraints_static_2} is satisfied and 
\begin{equation}
    H^j_t-(F+EK)S^j_t(F+EK)^T\geq 0,
\end{equation}
which by the Schur complement is equivalent to 
\begin{equation}\label{H_Lmi_beginning}
    \begin{pmatrix}
H^j_t &  &  F+EK\\
F^T+K^TE^T & &  P^j_t
\end{pmatrix}\geq 0.
\end{equation}
Since $G$ is full rank, multiplying \eqref{H_Lmi_beginning} by $\begin{pmatrix}
    \mathbf{I} & 0 \\ 0 & G^T
\end{pmatrix}$ from the left and by its transpose from he right, \eqref{H_Lmi_beginning} is equivalent to 
\begin{equation}\label{H_LMI_2_c}
    \begin{pmatrix}
H^j_t &   & FG+EL \\
G^T F^T+L^T E^T &  & G^TP^j_tG 
\end{pmatrix}\geq 0,
\end{equation}
That means that \eqref{H_LMI_2_c} with $H^j_t$ satisfying \eqref{LMI_constraints_static_2} is equivalent to $\mathcal{P}^j_t \subset \mathbb{X}$, $\forall (j,t) \in I_{\llbracket 0,N-1 \rrbracket}$.
From \eqref{G_eqn_3}, $ G^TP^j_tG \geq G^T+G-S^j_t$. Therefore, \eqref{LMI_constraints_static_1} is sufficient for \eqref{H_LMI_2_c}. As a result, \eqref{LMI_constraints_static_1}, \eqref{LMI_constraints_static_2} are sufficient for $\mathcal{P}^j_t \subset \mathbb{X}$, $\forall (j,t) \in I_{\llbracket 0,N-1 \rrbracket}$.
Moreover, due to Theorem \ref{Theorem_stability_static_suff}, \eqref{LMI_1_static} and \eqref{LMI_2_static} are sufficient for \eqref{P_static_1} and \eqref{P_static_2}. From  \eqref{P_static_1} and \eqref{P_static_2}, if $x_0 \in \mathcal{P}_0$ then for the system $\mathcal{S}_\mathcal{D}$, $x^j_t \in \mathcal{P}^j_{t}\subset \mathbb{X}$, $\forall (j,t)\in I_{\llbracket 0,N-1\rrbracket}$ and $x^j_N \in \mathcal{P}_0$, $\forall (j,N) \in I_N$. Hence the set $\mathcal{P}_0$ is robust periodic invariant for the system $\mathcal{S}_\mathcal{D}$. For the system $\mathcal{S}$, $x_t \in \mathbb{X}$, $\forall t \in \{0,1,\dots,N-1 \}$ due to Lemma \ref{convexity_lemma}.A. Furthermore due to this and Lemma \ref{convexity_lemma}.B the set $\mathcal{P}_0$ is also robust periodic invariant for the system $\mathcal{S}$, which completes the proof.
\end{proof}

\textbf{Proof of Corollary \ref{FPI_corollary_static}:}
\begin{proof}
Since $ \mathcal{P}^j_t\subset \mathbb{X}$, $\forall (j,t) \in I_{ \llbracket 0,N-1 \rrbracket}$, and $\mathbb{X}$ is convex, therefore, 
\begin{equation*}
\bar{\mathbb{P}}_t \subset \mathbb{X}, \textbf{ } t \in \{0,1,\dots,N-1\}
\end{equation*}
If $x_t \in \bar{\mathbb{P}}_t=Co(\{ \mathcal{P}^j_t \}, \textbf{ }j \in \{1,2,\dots,n_d^{t} \})$ then there exist vectors $x^j_t \in \mathcal{P}^j_t$ and scalars $\beta^j_t\geq0$ such that $\sum^{n_d^t}_{j=1}\beta^j_t=1$ and $x_t=\sum^{n_d^t}_{j=1}\beta^j_t x^j_t$. It follows that 
\begin{align}
x_{t+1}=&(A_t+B_t K)x_t \\=&\sum^{n_d^t}_{j=1}\beta^j_t(A_t+B_t K) x^j_t\\
     =&\sum^{n_d^t}_{j=1}\beta^j_t \sum^{n_d}_{i=1} \alpha_{t,i}(\bar{A}_i+\bar{B}_i K) x^j_t\\=&\sum^{n_d^{t+1}}_{l=1}\beta^l_{t+1}  x^l_{t+1}, 
\end{align}
where $x^l_{t+1}=(\bar{A}_i+\bar{B}_i K) x^j_t$. Due to \eqref{P_static_1}, $x^l_{t+1}\in \mathcal{P}^l_{t+1}$, if $t \in \{0,1,\dots,N-2\}$, and due to \eqref{P_static_2}, $x^l_{t+1}\in \mathcal{P}_0$ if $t=N-1$. Therefore,
\begin{equation*}
x_{t+1} \in Co(\{ \mathcal{P}^l_{t+1} \}, l \in \{1,2,\dots,n_d^{t+1} \})=\bar{\mathbb{P}}_{t+1}
\end{equation*}
if $t \in \{0,1,\dots, N-2 \}$, and $x_{t+1} \in \mathcal{P}_0$ if $t=N-1$, which completes the proof. 
\end{proof}

\textbf{Proof of Theorem \ref{Theorem_FPI_periodic_K}:}
\begin{proof}
From Theorem \ref{suff_theorem_differnt_gains}, \eqref{LMI_1_diff_K} and \eqref{LMI_2_diff_K} imply robust exponential stability of the unconstrained closed-loop system. Therefore, it remains to show that \eqref{LMI_constraints_periodic_1} and \eqref{LMI_constraints_periodic_2} imply the robust periodic invariance of $\mathcal{P}_0$.

Using the Schur complement, \eqref{LMI_constraints_periodic_1} is equivalent to 
\begin{equation}
    H^j_t-(F S^j_t+E L^j_t)P^j_t (F S^j_t+E L^j_t)^T\geq 0,
\end{equation}
which is the same as
\begin{equation}
    H^j_t-(F S^j_t+E L^j_t)P^j_t S^j_t P^j_t (F S^j_t+E L^j_t)^T\geq 0.
\end{equation}
Therefore from the definition of $K^j_t$,
\begin{equation}\label{H_Ks_eqn}
    H^j_t-(F +E K^j_t)S^j_t (F +E K^j_t)^T\geq 0.
\end{equation}
As shown in the proof of Theorem 2.9 in \cite{Kouvaritakis2016}, \eqref{LMI_constraints_periodic_2} and \eqref{H_Ks_eqn} are equivalent to  $\mathcal{P}^j_t= \{x| \textbf{ }x^T P^j_t x \leq 1 \} \subset \mathbb{X}^j_t$, $(j,t) \in I_{\llbracket 0,N-1 \rrbracket}$.
Moreover, due to Theorem \ref{suff_theorem_differnt_gains},  \eqref{LMI_1_diff_K} and \eqref{LMI_2_diff_K} are equivalent to \eqref{P_memory_1} and \eqref{P_memory_2}. From  \eqref{P_memory_1} and \eqref{P_memory_2}, if $x_0 \in \mathcal{P}_0$ then for the system $\mathcal{S}_\mathcal{D}$, $ x^j_t \in \mathcal{P}^j_{t}\subset \mathbb{X}^j_t$, i.e.,
\begin{equation}\label{F+EK_to_be_multiplied}
   (F+EK^j_t)x^j_t \leq \textbf{1},
\end{equation}
$\forall (j,t)\in I_{\llbracket 0,N-1\rrbracket}$ and $x^j_N \in \mathcal{P}_0$, $\forall (j,N) \in I_N$. Hence, $\mathcal{P}_0$ is robust periodic invariant for $\mathcal{S}_\mathcal{D}$. For the system $\mathcal{S}$, 
\begin{align}
     Fx_t+Eu_t=& F\sum^{n_d^t}_{j=1}\beta^j_tx^j_t+E\sum^{n_d^t}_{j=1}\beta^j_tK^j_tx^j_t,\\
               =& \sum^{n_d^t}_{j=1}\beta^j_t(F+EK^j_t)x^j_t, \label{cons_2}
\end{align}
$\forall (j,t)\in I_{\llbracket 0,N-1\rrbracket}$, for all $\beta^j_t$ that satisfy Definition \ref{beta_definition}. Multiplying \eqref{F+EK_to_be_multiplied} by $\beta^j_t$ and summing over all $j$, and from \eqref{cons_2}, we see that $ Fx_t+Eu_t\leq \textbf{1}$, $\forall t \in \{0,1,\dots,N-1 \}$. Furthermore, due to Lemma \ref{convexity_lemma}, $x_N \in \mathcal{P}_0$ and the set $\mathcal{P}_0$ is also robust periodic invariant for the system $\mathcal{S}$.
\end{proof}

\textbf{Proof of Corollary \ref{Corollary_FPI_diff_K}:}
\begin{proof}
	Since $ \mathcal{P}^j_t\subset \mathbb{X}$, $\forall (j,t) \in I_{ \llbracket 0,N-1 \rrbracket}$, and $\mathbb{X}$ is convex, therefore, 
	\begin{equation*}
		\bar{\mathbb{P}}_t \subset \mathbb{X}, \textbf{ } t \in \{0,1,\dots,N-1\}
	\end{equation*}
	If $x_t \in \bar{\mathbb{P}}_t=Co(\{ \mathcal{P}^j_t \}, \textbf{ }j \in \{1,2,\dots,n_d^{t} \})$ then there exist vectors $x^j_t \in \mathcal{P}^j_t$ and scalars $\beta^j_t\geq0$ such that $\sum^{n_d^t}_{j=1}\beta^j_t=1$ and $x_t=\sum^{n_d^t}_{j=1}\beta^j_t x^j_t$. Using the LITPC specified by Theorem \ref{Theorem_FPI_periodic_K}, it follows that 
	\begin{align}
		x_{t+1}=&A_tx_t+B_tu_t\\
		=&A_t\sum^{n_d^t}_{j=1}\beta^j_t x^j_t+B_t \sum^{n_d^t}_{j=1}\beta^j_t K^j_tx^j_t \\=&\sum^{n_d^t}_{j=1}\beta^j_t(A_t+B_t K^j_t) x^j_t\\
		=&\sum^{n_d^t}_{j=1}\beta^j_t \sum^{n_d}_{i=1} \alpha_{t,i}(\bar{A}_i+\bar{B}_i K^j_t) x^j_t\\=&\sum^{n_d^{t+1}}_{l=1}\beta^l_{t+1}  x^l_{t+1}, 
	\end{align}
	where $x^l_{t+1}=(\bar{A}_i+\bar{B}_i K) x^j_t$. Due to \eqref{P_memory_1}, $x^l_{t+1}\in \mathcal{P}^l_{t+1}$, if $t \in \{0,1,\dots,N-2\}$, and due to \eqref{P_memory_2}, $x^l_{t+1}\in \mathcal{P}_0$ if $t=N-1$. Therefore,
	\begin{equation*}
		x_{t+1} \in Co(\{ \mathcal{P}^l_{t+1} \}, l \in \{1,2,\dots,n_d^{t+1} \})=\bar{\mathbb{P}}_{t+1}
	\end{equation*}
	if $t \in \{0,1,\dots, N-2 \}$, and $x_{t+1} \in \mathcal{P}_0$ if $t=N-1$, which completes the proof. 
\end{proof}

\section{Important Auxiliary result}\label{Appendix_sec_Auxilary}
The following Lemma is an auxiliary result that is crucial for proving the necessity part of Lemma \ref{prop_suff_FSLF_K_different}. 
\begin{lemma}\label{essential_lemma}
Let $P_0>0$, $M_i\geq 0$ be $n_x \times n_x$ symmetric matrices $\forall i \in \{1,2,\dots,s \}$, $s \in \mathbb{Z}_{\geq 1}$. If $P_0-A^T M_i A >0 $, $\forall i \in  \{1,2,\dots,s \}$, where $A \in \mathbb{R}^{n_x \times n_x}$, then there exists a symmetric matrix $Q$ such that $Q>M_i$, $\forall i \in  \{1,2,\dots,s \}$ and $P_0-A^T Q A >0$.
\end{lemma}
\begin{proof}
    The proof is constructive in the sense that two different constructions of $Q$ matrices that satisfy the Lemma, for the case when $A$ is full rank and for the case when $A$ is not full rank result.\\
\textbf{Case 1 ($A$ is full rank):}\\
There exists a sufficiently small constant $\mu>0$ such that $P_0-A^T M_i A -\mu \mathbf{I}>0 $, $\forall i \in  \{1,2,\dots,s \}$. Therefore, $M_i<A^{-T}(P_0-\mu \mathbf{I})A^{-1}$. Hence, the Lemma holds with $Q=A^{-T}(P_0-\mu \mathbf{I})A^{-1}$.\\
\textbf{Case 2 ($A$ is not full rank):}\\
Let the singular value decomposition of $A$ be
\begin{equation*}
A=U\Sigma V^T=\begin{pmatrix}
U_1 & U_2 
\end{pmatrix}     \begin{pmatrix}
\Sigma_1 & 0 \\ 0 & 0
\end{pmatrix}  \begin{pmatrix}
V^T_1 \\ V^T_2
\end{pmatrix}=U_1 \Sigma_1 V^T_1.      
\end{equation*}
Factorize each of the matrices $M_i$ as follows:
\begin{equation}\label{M_factorize}
M_i=U\bar{M}_i U^T=\begin{pmatrix}
U_1 & U_2 
\end{pmatrix}     \begin{pmatrix}
M_i^{11} & M_i^{12} \\ M_i^{12^T}  & M_i^{22} 
\end{pmatrix}  \begin{pmatrix}
U^T_1 \\ U^T_2
\end{pmatrix}.      
\end{equation}
Multiplying the inequality $P_0-A^T M_i A >0$, $\forall i \in \{1,2,\dots,s \}$ by $V^T$ from the left and by $V$ from the right, we have
\begin{equation}
\begin{pmatrix}
V^T_1 P_0 V_1-\Sigma_1 U^T_1 M_i U_1 \Sigma_1 & \textbf{ } & V_1^T P_0 V_2 \\ V_2^T P_0 V_1  & \textbf{ } & V_2^T P_0 V_2 
\end{pmatrix}  > 0,   
\end{equation}
$\forall i \in   \{1,2,\dots,s \}$, which by the Schur complement is equivalent to
\begin{equation}\label{tobe_rear}
    V^T_1 P_0 V_1-\Sigma_1 U^T_1 M_i U_1 \Sigma_1-  V_1^T P_0 V_2 (V_2^T P_0 V_2)^{-1} V_2^T P_0 V_1>0,
\end{equation}
$\forall i \in  \{1,2,\dots,s \}$. Multiplying \eqref{tobe_rear} by $\Sigma_1^{-1}$ from the left and from the right and rearranging, we have 
\begin{equation}\label{rearranged}
 U^T_1 M_i U_1 < \Sigma_1^{-1}(V^T_1 P_0 V_1-  V_1^T P_0 V_2 (V_2^T P_0 V_2)^{-1} V_2^T P_0 V_1)\Sigma_1^{-1},
\end{equation}
$\forall i \in  \{1,2,\dots,s \}$. Since inequality \eqref{rearranged} is strict, therefore there exists a sufficiently small $\mu>0$ such that $\forall i \in  \{1,2,\dots,s \}$ 
\begin{equation}\label{rearranged_2}
 U^T_1 M_i U_1 < \Sigma_1^{-1}(V^T_1 P_0 V_1-  V_1^T P_0 V_2 (V_2^T P_0 V_2)^{-1} V_2^T P_0 V_1 -\mu \mathbf{I})\Sigma_1^{-1}.
\end{equation}
Define $Q$ as 
\begin{equation}\label{Q_factorize}
Q=U\bar{Q} U^T=\begin{pmatrix}
U_1 & U_2 
\end{pmatrix}     \begin{pmatrix}
Q^{11} & 0 \\ 0  & Q^{22} 
\end{pmatrix}  \begin{pmatrix}
U^T_1 \\ U^T_2
\end{pmatrix},
\end{equation}
with
\begin{equation}\label{Q11_eq}
Q^{11}= \Sigma_1^{-1}(V^T_1 P_0 V_1-  V_1^T P_0 V_2 (V_2^T P_0 V_2)^{-1} V_2^T P_0 V_1 -\mu \mathbf{I})\Sigma_1^{-1}    
\end{equation}
Note that from \eqref{M_factorize}, $U^T_1 M_i U_1=M^{11}_i$, and from \eqref{Q_factorize}, $U^T_1 Q U_1=Q^{11}$, and therefore, from \eqref{rearranged_2} and \eqref{Q11_eq} we conclude that  $Q^{11}-M_i^{11}>0$, $\forall i \in  \{1,2,\dots,s \}$. Note that $M^{11}_i\geq 0$, which implies that $Q^{11}>0$. Define $Q^{22}=c \mathbf{I}$, where $c>0$ is chosen sufficiently large such that $\forall i \in  \{1,2,\dots,s \}$
\begin{equation}\label{c_eqn}
    c \mathbf{I} > M_i^{22}+M_i^{12^T}(Q_{11}-M_i^{11})^{-1}M^{12}_i.
\end{equation}
Since $U$ is full rank, $Q^{11}>0$ and $Q^{22}>0$, therefore $Q>0$. Since $V$ is full rank, we have that $P-A^T Q A>0$ if and only if $V^T (P_0-A^T Q A)V>0$. Using the left hand side of the latter 
\begin{equation}\label{Big_equation}
    V^T (P_0-A^T Q A)V= \begin{pmatrix}
V^T_1 P_0 V_1-\Sigma_1 U^T_1 Q U_1 \Sigma_1 &  V_1^T P_0 V_2 \\ V_2^T P_0 V_1  &  V_2^T P_0 V_2
\end{pmatrix}.
\end{equation}
Using \eqref{Q_factorize} and \eqref{Q11_eq} in \eqref{Big_equation}, we have
\begin{equation}\label{Big_equation_2}
   V^T (P_0-A^T Q A)V= \begin{pmatrix}
Z &  V_1^T P_0 V_2 \\ V_2^T P_0 V_1  &  V_2^T P_0 V_2 
\end{pmatrix}, 
\end{equation}
where $Z=V_1^T P_0 V_2 (V_2^T P_0 V_2)^{-1} V_2^T P_0 V_1 +\mu \mathbf{I}$.
The Schur complement of the right hand side of \eqref{Big_equation_2} is equal to $Z- V_1^T P_0 V_2 (V_2^T P_0 V_2)^{-1} V_2^T P_0 V_2 =\mu \mathbf{I}>0$. Therefore, $ V^T (P_0-A^T Q A)V>0$, and as a result $P_0-A^T Q A>0$. It remains to show that this choice of $Q$, satisfies $Q>M_i$, $\forall i \in  \{1,2,\dots,s \}$. Since $U$ is full rank, $Q>M_i$ is equivalent to $\bar{Q}>\bar{M}_i$, $\forall i \in  \{1,2,\dots,s \}$, and 
\begin{equation}\label{Q-M_eqn}
   \bar{Q}-\bar{M}_i= \begin{pmatrix}
Q^{11}-M^{11}_i &  -M^{12}_i \\ -M^{12^T}_i &  Q^{22}-M^{22}_i 
\end{pmatrix}, 
\end{equation}
$\forall i \in  \{1,2,\dots,s \}$. The Schur complement of the right hand side of \eqref{Q-M_eqn} (with respect to the right bottom block) is  $Q^{22}-M^{22}_i-M^{12^T}_i (Q^{11}-M^{11}_i)^{-1}M^{12}_i$. Since $Q^{22}=c \mathbf{I}$ and using \eqref{c_eqn}, we have $Q^{22}-M^{22}_i-M^{12^T}_i (Q^{11}-M^{11}_i)^{-1}M^{12}_i>0$, and hence $\bar{Q}-\bar{M}_i>0$, which implies that $Q>M_i$ $\forall i \in  \{1,2,\dots,s \}$. 
\end{proof}
\end{document}